%% file: main.tex
\documentclass[letterpaper,twocolumn,10pt]{article}
\usepackage{style}

\usepackage{amsmath,amssymb,amstext} 
\usepackage[pdftex]{graphicx} 
\usepackage[utf8]{inputenc}
\usepackage{amsthm}
\usepackage{comment}
\usepackage{multirow}
\usepackage{pifont}
\usepackage{xifthen}
\usepackage[underline=false]{pgf-umlsd}
\usepackage{tikz}
\usepackage{svg}
\usepackage{algorithm}
\usepackage[noend]{algpseudocode}
\usepackage{algpseudocode}
\usepackage{enumitem}
\usepackage[lambda ,advantage, operators, sets , adversary, landau, probability, notions, logic, ff , mm, primitives, events, complexity, asymptotics, keys]{cryptocode}
\usepackage{booktabs} 
\usepackage{comment}
\usepackage{xfrac}
\usepackage[normalem]{ulem} 
\usepackage{graphicx}
\usepackage{makecell}
\usepackage{boldline}
\usepackage{stmaryrd}
\usepackage{relsize}
\usepackage{array}
\usepackage{stfloats}
\usepackage{xurl}
\usepackage[numbers,sort]{natbib}
\newcolumntype{?}{!{\vrule width 1pt}}
\usepackage{tikz}
\usepackage{pgfplotstable}
\usepackage{pgfplots}
\pgfplotsset{compat=1.17}
\usepackage{subfig}
\usepackage{amsthm,thm-restate}
\usepgfplotslibrary{fillbetween}
\microtypecontext{spacing=nonfrench}

\usepackage{cleveref}

\input{macros}

\begin{document}

\date{}

\title{\Large \bf Constant-weight PIR: Single-round Keyword PIR via Constant-weight Equality Operators }

\author{
{\rm Rasoul Akhavan Mahdavi}\\
University of Waterloo\\
rasoul.akhavan.mahdavi@uwaterloo.ca
\and
{\rm Florian Kerschbaum}\\
University of Waterloo\\
florian.kerschbaum@uwaterloo.ca
} 


\maketitle

\begin{abstract}

Equality operators are an essential building block in tasks over secure computation such as private information retrieval. In private information retrieval (PIR), a user queries a database such that the server does not learn which element is queried. In this work, we propose \emph{equality operators for constant-weight codewords}. A constant-weight code is a collection of codewords that share the same Hamming weight. Constant-weight equality operators have a multiplicative depth that depends only on the Hamming weight of the code, not the bit-length of the elements. In our experiments, we show how these equality operators are up to 10 times faster than existing equality operators. Furthermore, we propose PIR using the constant-weight equality operator or \emph{constant-weight PIR}, which is a PIR protocol using an approach previously deemed impractical. We show that for private retrieval of large, streaming data, constant-weight PIR has a smaller communication complexity and lower runtime compared to SEALPIR and MulPIR, respectively, which are two state-of-the-art solutions for PIR. Moreover, we show how constant-weight PIR can be extended to keyword PIR. In keyword PIR, the desired element is retrieved by a unique identifier pertaining to the sought item, e.g., the name of a file. Previous solutions to keyword PIR require one or multiple rounds of communication to reduce the problem to normal PIR. We show that constant-weight PIR is the first practical single-round solution to  single-server keyword PIR.

\end{abstract}

\section{Introduction}

Homomorphic encryption permits computation on encrypted data without the need to decrypt. For example, some homomorphic encryption schemes allow addition and multiplication~\cite{fan2012somewhat,10.1145/2090236.2090262, 10.1007/978-3-642-13190-5_2}, from which arbitrary functions are derived. However, operations with homomorphic encryption are not equally expensive and multiplications are up to 20 times more expensive compared to additions. Furthermore, the maximum number of sequential multiplications, i.e., the multiplicative depth, that can be performed is limited by the parameters of the encryption scheme. Hence, it is beneficial to derive functions using a smaller multiplicative depth and fewer multiplications.
Equality operators are an important function used in many applications~\cite{Akavia_Feldman_Shaul_2019, Akavia2019, kacsmar2020differentially, 10.1145/2517488.2517497,10.1145/3411495.3421361}.
However, the cost of performing one equality check using homomorphic encryption is often impractical due to the high multiplicative depth of the equality circuit. Specifically, the multiplicative depth of existing equality operators depends on the bit-length of the elements, which limits scalability.

In this work, we propose \emph{equality operators for constant-weight codewords} as a new, efficient way to compare homomorphically encrypted data. Constant-weight codewords are binary strings that share the same Hamming weight. We design equality operators specifically for these codewords with a multiplicative depth that depends only on the Hamming weight of the code, not the bit-length of elements. Our experiments show equality operators for constant-weight codewords are up to 10$\times$ faster than existing equality operators.

Private Information Retrieval (PIR), first introduced by Chor et al.~\cite{chor1995private}, is an example of an application in which equality operators play a crucial role. In a PIR protocol, a user retrieves an element from a database, such that the database does not learn which element is retrieved. Typically, elements are retrieved using the physical address of the desired item, which we call \emph{index PIR}. In another variant called keyword PIR~\cite{chor1997private}, the user's desired element is retrieved using an identifier pertaining to the sought item, e.g., the name of a file. State-of-the-art solutions for keyword PIR reduce it to index PIR using one or multiple extra rounds of communication~\cite{chor1997private}. Ali et al.~\cite{Ali2019CommunicationComputationTI} propose a probabilistic hashing technique to map identifiers from a large domain to a small table.

In this work, we propose PIR using constant-weight equality operators or \emph{constant-weight PIR}. Our protocol uses an approach that was assumed to be impractical, which differs from that of related work, specifically SEALPIR~\cite{sealpir} and MulPIR~\cite{Ali2019CommunicationComputationTI}, two efficient PIR protocols.
Constant-weight PIR also scales to databases with large payload data or streaming data with less communication and computation overhead compared to SEALPIR and MulPIR, respectively. For example, we show that for 16000 rows, the runtime of MulPIR grows twice as fast as constant-weight PIR, as a function of the payload size. Consequently, constant-weight PIR has a smaller runtime than MulPIR when the payload size exceeds 268 KB, which corresponds to a database size of 4.3 GB. Similarly, the communication complexity of SEALPIR grows twice as fast as constant-weight PIR as a function of the payload size.

Moreover, due to the modularity and simplicity, it can also be extended to keyword PIR with minor modification, no extra rounds, and minimal overhead. Constant-weight keyword PIR is the first efficient single-round solution for single-server keyword PIR.


Single-round single-server keyword PIR is useful for the application of private file retrieval~\cite{Mayberry2014EfficientPF}.
Compared to existing approaches, PIR has asymptotically optimal communication overhead for privately retrieving large items from a database~\cite{Mayberry2014EfficientPF}.
Constant-weight PIR in particular also has a small computational overhead, compared to other PIR protocols, when the retrieved item is large and also allows updates in the database with no interaction with the users.

We show through our experiments how the size of the domain of keywords only affects one of the three steps performed by the server in constant-weight PIR. Hence, the domain of keywords can be expanded with marginal cost.
We also show how the constant-weight code used in our protocol is a more space-efficient representation of a PIR query, for a fixed multiplicative depth, compared to existing work. Specifically, for a multiplicative of $d$ in the PIR protocol over a database with $n$ possible identifiers, the representation used in constant-weight PIR has a size of $O(\sqrt[2^d]{n})$. In contrast, SEALPIR and MulPIR use a representation for the query of size $O(d\sqrt[d+1]{n})$.

Overall, the contributions of this paper are as follows:
\begin{itemize}
\itemsep0mm
    \item Novel equality operators for constant-weight codewords
    \item PIR using constant-weight equality operators
    \item Experimental evaluation of the equality operators
    \item Evaluation of constant-weight PIR and comparison with existing index PIR protocols
    \item Detailed analysis of constant-weight keyword PIR
\end{itemize}

\input{background}

\input{constructions}

\input{evaluation}

\input{conclusion}

\section*{Acknowledgements}
We would like to thank Ian Goldberg for his useful comments on an earlier version of this work. We also thank our reviewers for their comments and particularly our shepherd, Tancrède Lepoint, which provided helpful insights and suggestions to clarify our contributions.
This work benefited from the use of the CrySP RIPPLE Facility at the University of Waterloo.

\bibliographystyle{plain}
\bibliography{main}

\appendix
\input{appendix}

\end{document}

%% file: macros.tex
\newtheorem{theorem}{Theorem}
\newtheorem{definition}{Definition}

\newcounter{margin} 
\definecolor{Dgreen}{rgb}{0, 0.6, 0}

\newcommand{\sdend}{%
   \begin{pgfonlayer}{umlsd@background}%
      \ifnum\value{instnum}>0%
         \foreach \t [evaluate=\t] in {1,...,\theinstnum}{%
            \draw[dotted] (inst\t) -- ++(0,-\theseqlevel*\unitfactor-2.2*\unitfactor);%
         }%
         
      \fi%
      \ifnum\value{threadnum}>0%
         \foreach \t [evaluate=\t] in {1,...,\thethreadnum}{%
            \path (thread\t)+(0,-\theseqlevel*\unitfactor-0.1*\unitfactor) node (threadend) {};%
            \tikzstyle{threadstyle}+=[threadcolor\t]%
            \drawthread{thread\t}{threadend}%
         }%
      \fi%
   \end{pgfonlayer}%
}

\newcommand{\formatcomment}[1]{\normalsize\textcolor{blue!25!black}{\texttt{#1}}}

\algrenewcommand{\algorithmiccomment}[1]{\hfill\parbox{5.2cm}{\formatcomment{//\,#1}}}
\algrenewcommand{\algorithmicprocedure}{\textbf{algorithm}}
\algnewcommand{\algorithmicforeach}{\textbf{for each}}
\algdef{S}[FOR]{ForEach}[1]{\algorithmicforeach\ #1\ \algorithmicdo}

\algdef{SE}[FUNCTION]{Function}{EndFunction}%
   [2]{\algorithmicfunction\ \textproc{#1}\ifthenelse{\equal{#2}{}}{}{(#2)}}%
   {\algorithmicend\ \algorithmicfunction}%

\newcommand{\polydegree}[0]{N}

\newcommand{\db}[0]{\mathbb{DB}}
\newcommand{\dbsize}[0]{n}
\newcommand{\binaryset}[0]{{\{0,1\}}^{\ell}}

%% file: background.tex
\section{Background and Related Work}
\label{background}

\subsection{Homomorphic Encryption}

Homomorphic Encryption allows computation on encrypted data, without the need for decryption or access to the secret key. This maintains the secrecy of the data while computation is performed. One use case is a client delegating computation on its data to a remote, untrusted server.

The concept of homomorphic encryption was introduced by Rivest et al.~\cite{rivest1978data}.
In 2009, Gentry proved the existence of a \emph{fully homomorphic} cryptosystem based on lattices that can evaluate arbitrary functions on encrypted data~\cite{gentry2009fully}.

Multiple lattice-based cryptosystems were proposed following the seminal work of Gentry which improved the efficiency drastically~\cite{10.1007/978-3-642-13190-5_2,10.1145/2090236.2090262,10.1007/978-3-642-45239-0_4,10.1007/978-3-642-40041-4_5}. 
Many homomorphic cryptosystems are used in a \emph{leveled} fashion.  A leveled homomorphic cryptosystem allows only a predefined number of sequential multiplications, determined by the parameters of the cryptosystem. The Fan–Vercauteren cryptosystem is an example that we explain in the next subsection.

\subsubsection{Fan–Vercauteren (FV) Cryptosystem.}
The Fan–Vercauteren cryptosystem~\cite{fan2012somewhat} is a lattice-based cryptosystem where plaintexts are elements from the polynomial ring $R_t = \ZZ_t[x]/(x^N+1)$.
The \emph{polynomial modulus degree}, $N$, is a power of two and $t$ is the \emph{plaintext modulus}.
Messages must be encoded as a polynomial in the field before they can be encrypted.
An FV ciphertext is an array of polynomials, each from $R_q = \ZZ_q[x]/(x^N+1)$, where $q$ is called the \emph{coefficient modulus}. In the simplest case, the ciphertext is only two polynomials. Let $\mathcal{C}$ denote the ciphertext space. $N$ and $q$ determine both the security parameter and how many homomorphic operations can be performed on ciphertexts before decryption is necessary.

In addition to the standard operations for a cryptosystem, i.e., key generation, encryption and decryption, FV supports homomorphic operations over the ring as well. Four of these operations are listed below. All operations over plaintexts are in the ring $R_t$.

\begin{itemize}
\itemsep-1mm
    \item \textbf{Addition:} Given ciphertexts $c_1(x), c_2(x)\in \mathcal{C}$ that encrypt $m_1(x), m_2(x) \in R_t$, respectively, output $c_A(x)$ which encrypts $m_1(x) + m_2(x)$.
    \item \textbf{Plain Multiplication:} Given $m_1(x) \in R_t$ and $c_2(x)\in \mathcal{C}$ that encrypts $m_2(x)\in R_t$, output $c_{PM}(x)$ which encrypts $m_1(x)m_2(x)$.
    \item \textbf{Multiplication:} Given ciphertexts $c_1(x), c_2(x) \in \mathcal{C}$ that encrypt $m_1(x), m_2(x) \in R_t$, respectively, output $c_M(x)$ which encrypts $m_1(x) m_2(x)$.
    \item \textbf{Substitution:} Given $c(x)\in \mathcal{C}$ that encrypts $m(x)$ and an integer $k$, output $c_{S}(x)$ which encrypts $m(x^k)$.
\end{itemize}

In the rest of this paper, $\texttt{PM}$ and $\texttt{M}$ denote plaintext multiplication and homomorphic multiplication, respectively.

\subsubsection{Microsoft SEAL Library}
The SEAL library~\cite{sealcrypto} implements the FV cryptosystem and supports all the operations mentioned above. Specifically, the implementation for the substitution operation in this library was first introduced by Angel et al.~\cite{sealpir} based on the plaintext slot permutation technique discussed by Gentry et al.~\cite{10.1007/978-3-642-29011-4_28}.
One FV plaintext can encode $\polydegree\log_2t$ bits of data. Also, the size of the smallest ciphertext that encrypts a plaintext is $2\polydegree\log_2 q$ bits. An important parameter is the \emph{expansion factor} which is the ratio between the size of a ciphertext and the largest plaintext that can be encrypted and is equal to $F = 2\log q /\log t$.
In the rest of this paper, $F$ denotes the expansion factor of the FV cryptosystem.
\Cref{tab:seal-ops} compares the four described operations in terms of speed and noise grown, as implemented in SEAL 3.6.

\begin{table}[ht]
\centering
\caption[Runtime cost of operations in SEAL 3.6]{Runtime cost of operations in SEAL 3.6, for $\polydegree\in\{2048, 4096, 8192, 16384\}$ and the default ciphertext modulus. * Time and noise growth in plain multiplication also depend on the value of the unencrypted operand.}
\label{tab:seal-ops}
\resizebox{\columnwidth}{!}{%
\begin{tabular}{c|c|c|c|c|c} \toprule
\multirow{2}{*}{\textbf{Operation}} & \multicolumn{4}{c|}{\textbf{Time ($\mu s$)}} & \multirow{2}{*}{\textbf{Noise Growth}}\\ \cline{2-5}
 & $\polydegree=2048$ & $\polydegree=4096$ & $\polydegree=8192$ & $\polydegree=16384$ & \\ \midrule
        Addition & 6 & 19 & 67 & 435 & Additive \\\hline
        Plain Mult.$^{*}$ & 12--135 & 30--529 & 105--2201 & 509--9647 & Multiplicative \\\hline
        Multiplication & - & 3823 & 15744 & 66908 & Multiplicative \\\hline
        Substitution & - & 768 & 4137 & 26047 & Additive \\\bottomrule
\end{tabular}
}
\end{table}

\subsection{Private Information Retrieval}

Private Information Retrieval (PIR)~\cite{chor1995private} is a protocol where a user retrieves an element from a database, such that the owner of the database cannot determine which element was retrieved. There are two forms of PIR protocols. In the first form, which we denote \emph{index PIR}, the user holds the \emph{physical} address of the item, e.g., the row in a database table or the index in a public registry. In the second form, called \emph{keyword PIR}, the physical address of the desired item may not be known and it is only accessible by an identifier pertaining to the sought item, e.g., the name of a file.

The privacy guarantee of a PIR protocol can be information-theoretic or computational. Information-theoretic PIR (IT-PIR) is private even in the presence of a computationally unbounded adversary~\cite{chor1995private,Beimel2002BreakingTO,raidpir,Ambainis2000}.
Computational PIR (CPIR) relaxes the assumption to an adversary with bounded computational power.
In the single-server setting, which is the focus of this paper, solutions rely on some intractability assumption, e.g., the hardness of determining the quadratic residuosity modulo composite numbers~\cite{goldwasser1984probabilistic,646125} or the security of lattice-based cryptosystems~\cite{aguilar2016xpir, sealpir, Ali2019CommunicationComputationTI, 10.1007/978-3-319-11203-9_22, tcc-2019-30003,10.1007/978-3-662-44774-1_16,6189348}.

In CPIR solutions, each item in the database has to be processed at least once, otherwise, it can be trivially excluded from the list of potential queries and compromise privacy. Sion and Carbunar argued that the time required for any single-server CPIR protocol would exceed the time required for the trivial solution of simply downloading the entire database~\cite{sionpir}. Later work by Aguilar-Melchor et al. showed this argument to be incorrect with the use of lattice-based cryptosystems, which have smaller per-bit computation cost when used in a batched fashion~\cite{aguilar2016xpir}. They showed that PIR is a faster than downloading the database over low-bandwidth networks. 

\subsection{Single-Server computational PIR}
\label{sec:single-server-cpir}

Single-server computational PIR solutions aim to perform better than the \emph{trivial} solution of downloading the entire database. In the trivial solution, the \emph{download cost} for the user is equal to the size of the database, with no \emph{upload cost} for the user. Downloading the entire database also comes at almost no computational burden for the server, i.e., the \emph{computational cost} is zero.
We compare single-server CPIR protocols based on the upload, download, and computational cost.

CPIR protocols utilizing homomorphic encryption are the most practical solutions to date~\cite{aguilar2016xpir, sealpir, Ali2019CommunicationComputationTI}. All these solutions expand on a baseline method that works as follows:

\paragraph{Baseline PIR method.} Let $\db$ denote the database with $n$ rows and $\db[i]$ denote the $i^{th}$ row in this database. Also, throughout this paper, define $[n]=\{0,1,...n-1\}$, for any $n\in\NN$. When the goal is to retrieve row $q$, a response $r_q$ is derived as 
\begin{align}
   r_q = \sum_{i\in[n]} \mathbb{I}(i = q) \cdot \db[i].
   \label{eq:baseline-pir}
\end{align}
where $\mathbb{I}(\cdot)$ denotes an indicator function which is one when the input evaluates to true and zero otherwise. It is easy to verify that if $q\in[n]$ then $r_q = \db[q]$. Equation (\ref{eq:baseline-pir}) is an inner product between the database and a vector of bits called the \emph{selection vector}. For obtaining element $q$ in the database, the selection vector is one in index $q$ and zero otherwise.

PIR protocols realizing Equation (\ref{eq:baseline-pir}) encrypt the bits of the selection vector with a homomorphic encryption scheme that supports addition and plaintext multiplication and perform the operations in Equation (\ref{eq:baseline-pir}) over ciphertexts. In XPIR~\cite{aguilar2016xpir} and SealPIR~\cite{sealpir}, two recent practical solutions, an additive homomorphic encryption scheme is used. MulPIR~\cite{Ali2019CommunicationComputationTI} is the first practical solution using a fully homomorphic encryption scheme, which is also the case for our work.

The server requires ciphertexts of the bits of the selection vector, i.e., $\mathbb{I}(i=q)$, to realize Equation (\ref{eq:baseline-pir}).
There are two general approaches for the server to acquire the encrypted bits of the selection vector: 1) Communicating the selection vector 2) Equality Operators.

In the first approach, the user generates the selection vector locally, encrypts it and transmits it to the server. XPIR, SealPIR, and MulPIR all take this approach. XPIR uploads the entire selection vector but provides experiments to show the practicality of this approach~\cite{aguilar2016xpir}. Despite its practicality, the upload cost of XPIR is on the order of the number of rows in the database which limits scalability.

Recursion is a method to reduce the upload cost to sublinear in the size of the database. It was first used by Kushilevitz and Ostrovsky~\cite{646125} and later Stern~\cite{10.1007/3-540-49649-1_28}. This approach is also used in SealPIR and MulPIR. In the next section, we describe how recursion is done in SealPIR, which is conceptually similar to prior work.

\subsubsection{SealPIR}
SealPIR~\cite{sealpir} is a PIR scheme based on the SEAL library which uses a \emph{query compression} technique and \emph{recursion} to reduce the upload cost. They also use additive homomorphic encryption in a layered fashion.

In SealPIR, to communicate fewer ciphertexts, the user encodes multiple bits into one plaintext, which is called the query compression technique. Specifically, for a selection vector $(s_i)_{i\in[\dbsize]}$, the user constructs the plaintext $p(x) = \sum_{i\in[\dbsize]}s_i x^i$ and encrypts it.  Recall that in SEAL, plaintexts are polynomials of degree at most $N$, so if the size of the selection vector exceeds the polynomial degree, $\ceil{n/\polydegree}$ ciphertexts are used.
As a consequence of the compression technique, SealPIR performs a novel \emph{oblivious expansion} on the server to extract a vector of ciphertexts such that each bit of the selection vector is in a separate ciphertext. SealPIR uses the substitution operation to perform the oblivious expansion. \Cref{alg:oblivious-expand-sealpir} depicts this procedure for expanding one ciphertext into a vector of $2^c$ ciphertexts, for $c \in \{0,1,...,\log_2 N\}$.

\begin{algorithm}[!ht]
	 \caption[]{\textsc{SealPIR Oblivious Expansion}}
	 \label{alg:oblivious-expand-sealpir}
	 	\vspace{-3mm}
	 \begin{flushleft}
		 \textbf{Input:} $ct(x) \in \mathcal{C}$, $c \in \{0,1,...,\log_2 N\}$
	 \end{flushleft}
	 	\vspace{-3mm}
	 \begin{algorithmic}[1]
        \State $cts \leftarrow [ct(x)]$
        \For {$a\in [c]$}
            \For {$b\in [2^a]$}
                \State $c_0 = cts[b]$
                \State $c_1 = x^{2^{-a}} \cdot c_0$
                \State $cts[b] = c_0 + \texttt{Sub}_{N/2^{a}+1}(c_0)$
                \State $cts[b+2^a] = c_1 + \texttt{Sub}_{N/2^{a}+1}(c_1)$
            \EndFor
        \EndFor
        \State $inv = (2^{-c} \mod t)$
        \For {$i \in [2^c]$}
            \State $cts[i] \leftarrow inv \cdot cts[i]$
        \EndFor
	 \end{algorithmic}
	 	\vspace{-3mm}
	 \begin{flushleft}
    	 \textbf{Output:}  $cts \in \mathcal{C}^{2^c}$\\
	 \end{flushleft}
	 	\vspace{-3mm}
\end{algorithm}

To further reduce the upload cost, SealPIR uses a technique called \emph{recursion} in which the database is restructured into a $d$-dimensional table.
The users query is translated into a coordinate in this $d$-dimensional table.
Then instead of one selection vector, $d$ selection vectors are sent to the server, one for each dimension. We refer to $d$ as the \emph{recursion level}. The total size of the query is at least $d\ceil{\sqrt[d]{\dbsize}}$ which is sublinear in $\dbsize$ for any $d\geq2$.

To calculate the response to the query using the selection vectors, $d$ inner products are performed in sequence.
In SealPIR, an additive homomorphic encryption scheme is used so the multiplication in the first inner product is performed as a plaintext multiplication. However, the subsequent multiplications are between ciphertext, which is not supported. To overcome this issue, one ciphertext is treated as a plaintext in the multiplication. This is referred to as \emph{layered encryption} and results in the size of the response multiplying by a factor of $F$ where $F$ is the expansion factor of the ciphertext. More generally, the size of the response is multiplied by a factor of $F^{d-1}$ for recursion level equal to $d$. Overall, SealPIR performs $\sum_{i=0}^{d-1} n^{\frac{d-i}{d}}F^i$ plaintext multiplications for recursion level $d\geq 1$ and expansion factor of $F$ for the ciphertext.

Ali et al. proposed three additional optimizations to SealPIR to reduce the upload and download cost~\cite{Ali2019CommunicationComputationTI}. These three optimizations are: compressing the uploaded ciphertexts by encrypting using the secret key instead of the public key, compressing the response ciphertexts using modulus switching, and a modified oblivious expansion to fit more bits into the one ciphertext. Throughout this paper, \emph{SealPIR} denotes this modified version of the protocol.

\subsubsection{MulPIR}
MulPIR~\cite{Ali2019CommunicationComputationTI} replaces the layered encryption in SealPIR with homomorphic multiplications. This reduces the download cost drastically compared to SealPIR. However, it comes at the cost of increased computation for the server since homomorphic multiplications are more expensive than plain multiplications and larger parameters are required to allow more homomorphic multiplications.
Overall, MulPIR performs $\dbsize$ plaintext multiplications and $\sum_{i=1}^{d-1} n^{\frac{d-i}{d}}$ homomorphic multiplications, for a recursion level $d\geq 1$.

In SealPIR, due to the expansion in the response, the server can not perform any post-processing on the output which is a disadvantage of the protocol.
Examples of post-processing include deriving functions of the user's query or conjunctive and disjunctive PIR queries.
In contrast to SealPIR, the output of the MulPIR protocol can be post-processed before being sent back to the user. Ali et al.~\cite{Ali2019CommunicationComputationTI} describe how to perform conjunctive and disjunctive queries using MulPIR.

\subsection{Equality Operators}
\label{sec:eq-op}

Checking the equality of two values is an integral step in many tasks over encrypted data such as secure search~\cite{Akavia2019,Akavia_Feldman_Shaul_2019}, secure pattern matching~\cite{10.1145/2517488.2517497,10.1145/3411495.3421361}, private set intersection~\cite{kacsmar2020differentially,10.1145/3133956.3134061}, and PIR~\cite{chor1995private}.

We define an \emph{equality operator} as follows.

\begin{definition}[Equality Operator]
A procedure $f$ is an equality operator over a domain $D$ if $\forall x,y\in D$, 
\begin{align}
	f(x,y)=
    \left\{
    	\begin{array}{ll}
    		1  & \mbox{if\ \ } x=y \\
    		0 & \mbox{o.w.}
    	\end{array}
    \right.
\label{eq:equality-circuit}    
\end{align}

\end{definition}

In this section, we define two equality operators over their respective domains and derive the multiplicative depth of a circuit implementing each one. More operators exist which are summarized in a previous version of this work~\cite{rasoul-mmath}. When working with an element $x\in\{0,1\}^{\ell}$, we treat it as a string of bits and refer to the bits of the string by indexing, i.e., $x[i]$ denotes the $i^{th}$ bit of $x$.

\paragraph{Arithmetic Folklore Equality Operator.} This operator is used to compare two numbers in binary format. For a domain $D=\binaryset$, define $f_{AF}$ as
\begin{align}
    f_{AF}(x,y) = \prod_{i=0}^{\ell-1} \left(1 - (x[i] - y[i])^{2}\right)
\end{align}
for $x,y \in \binaryset$.
This operator is correct when operating over any field such as $\ZZ_p$.
The multiplicative depth of a circuit realizing this operator is equal to $1 + \ceil{\log_2\ell}$, where $\ell$ is the bit-length of the operands. The arithmetic folklore operator is oblivious to both input operands. This is critical in some applications, e.g., comparing two encrypted or secret shared numbers.

When one operator is public, the arithmetic folklore equality operator can be modified to perform less operations with a smaller multiplicative depth. The modified operator is as follows.

\paragraph{Plain Folklore Equality Operator.} For a domain $D=\binaryset$, define $f_{PF}$ as
\begin{align}
    f_{PF}(x,y) = \prod_{y[i]=0} \left( 1-x[i] \right) \prod_{y[i]=1} x[i]
    \label{folklore-plain-eq}
\end{align}
for $x,y \in \binaryset$.
This operator depends on the public operand, which is $y$ in this case. 
The multiplicative depth of a circuit realizing this operator is equal to $\ceil{\log_2\ell}$, where $\ell$ is the bit-length of the operands.

\subsubsection{PIR using Equality Operators}
As mentioned in \Cref{sec:single-server-cpir}, equality operators are another approach to PIR. In this approach, the user's query is encoded into some domain, encrypted and sent to the server. The server computes each bit of the selection vector, i.e., $\mathbb{I}(i=q)$, using an equality operator between the user's encrypted query and each identifier in the database. Then the server, using the encrypted bits of the selection vector, derives the encryption of $r_q$ using Equation (\ref{eq:baseline-pir}), which is then sent back to the user for decryption.

PIR with this approach using the folklore equality operator has the smallest upload cost amongst all non-trivial approaches. In this approach, only the optimal logarithmic binary encoding of the query is encrypted and uploaded. However, the computation cost is prohibitively high due to the multiplicative depth of the folklore equality circuit which depends on the number of rows in the database.
In general, PIR using equality operators is assumed to be impractical due to the high multiplicative depth of equality circuits as parameters scale~\cite{sealpir, Ali2019CommunicationComputationTI}. This work challenges this assumption.

\subsection{Keyword PIR}
\label{sec:keyword-pir}
In keyword PIR, a user retrieves an element from a database using a keyword or identifier pertaining to the sought item.
Another way to phrase this is that in index PIR, all addresses correspond to an element in the database, whereas in keyword PIR, some keywords may not correspond to any element.
Note that keyword PIR implies that the user does not know which keywords are present in the database. Otherwise, the user can simply refer to its desired keyword by its position in the list of sorted keywords.

Previous work has suggested solutions for keyword PIR which all basically reduce keyword PIR to index PIR.
Chor et al. suggested two solutions where the user interactively queries the server to privately obtain the physical address of the desired item, given the identifier~\cite{chor1997private}. With the physical address, the user then conducts index PIR to retrieve the sought item.
A common solution also proposed by Ali et al. involves a probabilistic hashing technique to map keywords into a small table such that index PIR is feasible~\cite{Ali2019CommunicationComputationTI}.

PIR using equality operators is another approach to keyword PIR, where the user's query is compared to all the keywords present in the database. However, since the cost of comparing keywords is prohibitively high for large keywords, this approach is assumed to be impractical. This work proposes a PIR protocol for index PIR which can be easily extended to keyword PIR with minimal change. Moreover, the practical computational cost of the constant-weight equality operator results in a practical keyword PIR protocol.

%% file: constructions.tex
\section{Constructions for Constant-weight Codes}
\label{constructions}

In this section, we describe our constructions. First, we propose equality operators for constant-weight codewords. Then we describe efficient mappings from other domains to constant-weight codewords to facilitate the use of our proposed operator in other contexts. Finally, we explain PIR using constant-weight codewords in detail.
\paragraph{Constant-weight Code.} A \emph{constant-weight code}, or an \emph{m-of-n code}, is a form of error detecting code where all codewords share the same Hamming weight. A \emph{binary constant-weight code} has the additional condition that all codewords are binary strings. The one-hot (unary) code and the \emph{balanced code} are two examples of a binary constant-weight code. In a balanced code, the number of ones is equal to the number of zeros in all codewords. 

The length of a code is the maximum bit-length of its codewords and the size of the code is the number of distinct codewords. For a binary constant-weight code of length $m$ and Hamming weight of $k$, the size is $\binom{m}{k}$. For a fixed Hamming weight $k$, to have a binary constant-weight code with a size of at least $n$, we must choose the length, $m$, such that $\binom{m}{k} \geq n$. By one approximation, we have $m \in O\left(\sqrt[k]{k!n} + k\right)$.
We denote the binary constant-weight code with length $m$ and Hamming weight $k$ by $CW(m,k)$.

In all the constructions, $k$ and $m$ denote the Hamming weight and code length, respectively.

\subsection{Equality Operators for Constant-weight Codewords}
\label{sec:eq-constant-weight-code}

We propose two variants of the equality operator over constant-weight codewords in this section. A third construction over a binary field is given in a previous version of this work~\cite{rasoul-mmath}.

\paragraph{Plain Constant-weight Equality Operator.} For two constant-weight codewords $x,y\in CW(m,k)$, 
\begin{align}
    f_{PCW}(x,y) = \prod_{y[j]=1} x[j]
    \label{eq:plain-cw-eq}
\end{align}
is the \emph{plain equality operator}. This operator is oblivious to the first operand but depends on the second. A circuit realizing this operator performs $k$ multiplications with a multiplicative depth of $\ceil{\log_2 k}$.

\paragraph{Arithmetic Constant-weight Equality Operator.}
For two constant-weight codewords $x,y\in CW(m,k)$, 
\Cref{alg:arith-cw-eq} describes the \emph{arithmetic equality operator} over constant-weight codewords. \Cref{alg:arith-cw-eq} operates over any field in which $k!$ has a multiplicative inverse.

\begin{algorithm}[!ht]
	 \caption[]{\textsc{Arithmetic Constant-weight Equality Operator}}
	 \label{alg:arith-cw-eq}
	 	\vspace{-3mm}
	 \begin{flushleft}
		 \textbf{Input:} $x, y \in CW(m,k)$
	 \end{flushleft}
	 	\vspace{-3mm}
	 \begin{algorithmic}[1]
	 	\State $k' = \sum\limits_{i\in[m]}x[i] \cdot y[i]$
	 	\vspace{2mm}
	 	\State $e = \frac{1}{k!} \prod\limits_{i\in[k]} (k' - i)$
	 	\vspace{2mm}
	 \end{algorithmic}
	 	\vspace{-5mm}
	 \begin{flushleft}
    	 \textbf{Output:}  $e \in \{0,1\}$\\
	 \end{flushleft}
	 	\vspace{-3mm}
\end{algorithm}

\begin{theorem}
For $x, y \in CW(m,k)$, if $f_{ACW}(x,y)$ is the output of \Cref{alg:arith-cw-eq}, then $f_{ACW}(x,y)$ is an equality operator.
\end{theorem}
\begin{proof}
If $x$ and $y$ are equal, the position of bits equal to one in their encodings are identical, and consequently, the inner product, $k'$, will be equal to $k$. When they are not equal, the inner product will be in the set $\{0,1,...,k-1\}$. Also, based on the definition of $e$ on line 2 of \Cref{alg:arith-cw-eq}, it holds that
\begin{align}
	e=
    \left\{
    	\begin{array}{ll}
    		1 & k' = k \\
    		0 & k' \in \{0,1,...,k-1\}
    	\end{array}
    \right.
\end{align}

Putting these two together, $e$ will be one, if and only if $x$ and $y$ are equal and zero otherwise.
\end{proof}

A circuit realizing this operator performs $m+k$ multiplications with a multiplicative depth of $1+\ceil{\log_2 k}$.

\subsection{Mappings to Constant-weight Codewords}
\label{sec:mappings}

The domain of all the operators described in this section is a constant-weight code. To benefit from these constructions in a setting where we want to compare elements from other domains, we also propose efficient mappings from other domains to constant-weight codewords. The goal is for the mapping (and inverse mapping) procedure to be efficient and less expensive than storing an equivalence table. We describe the perfect mapping below and detail the inverse perfect mapping and the lossy mapping in \Cref{sec:more-mappings}.

\paragraph{Perfect Mapping.}
This mapping is used to map numbers in the set $[n]$ to $CW(m,k)$ such that it is injective and has an inverse. To have the injective property, the code size must be at least $n$, i.e., $|CW(m,k)| = \binom{m}{k} \geq n$. The mapping procedure is given in \Cref{alg:mapping-perfect}.

\begin{algorithm}[!ht]
	 \caption[]{\textsc{Perfect Mapping}}
	 \label{alg:mapping-perfect}
	 	\vspace{-3mm}
	 \begin{flushleft}
		 \textbf{Input:} $x \in [n]$, $m,k \in \NN$ such that $\binom{m}{k} \geq n$ \vspace{-3mm}
	 \end{flushleft}
	 \begin{algorithmic}[1]
	 	\State $r= x$
	 	\State $h = k$
	 	\State $y= 0^m$
	 	\For {$m'  = m-1, ..., 1, 0 $}
		 	\If {$r \geq \binom{m'}{h}$} \\
				\hspace{10mm}$y[m']=1$ \\
				\hspace{10mm}$r = r - \binom{m'}{h}$ \\
				\hspace{10mm}$h = h-1$
			\EndIf
			\If {$h=0$}
			    break
			\EndIf
	 	\EndFor
	 \end{algorithmic}
	 	\vspace{-3mm}
	 \begin{flushleft}
    	 \textbf{Output:}  $y \in CW(m,k)$\\
	 \end{flushleft}
	 	\vspace{-3mm}
\end{algorithm}

Intuitively, this procedure is assigning the $i^{th}$ valid codeword from a sorted list of codewords to the number $i$. Creating this list and extracting the mapping corresponding to a number would be prohibitively expensive with an average complexity of $\theta\left(\binom{m}{k}\right)$. The complexity of our mapping procedure is $O(m+k)$.

\subsection{PIR using Constant-weight Codewords}
\label{pir-constant-weight-code}

In this section, we describe our protocol for PIR using constant-weight codewords, which we name \emph{constant-weight PIR}. Our protocol follows the approach using equality operators with the plain constant-weight equality operator at its core. It is the first practical and scalable PIR protocol using the equality operator approach.

The PIR protocol is conducted between a server and user. The server holds a database, $\mathbb{DB}$, with $n$ identifiers.
Each identifier corresponds to some payload data in the database.
We denote the set of identifiers in the database by $\mathbb{ID}$. The user holds a query $q$ from the domain of identifiers which we denote by $S(\mathbb{ID})$. We know by definition that $\mathbb{ID} \subseteq S(\mathbb{ID})$, but the user's query might not necessarily be in the database. Previous work, including SealPIR and MulPIR, focuses mainly on PIR when $|S(\mathbb{ID})|=|\mathbb{ID}|=n$, i.e., index PIR. In contrast, our work is applicable for both index and keyword PIR.
We first describe constant-weight PIR for index PIR and explain how to expand our construction to keyword PIR in \Cref{sec:pir-sparse}.

The protocol consists of four main stages: Setup, Query, Process, and Extract. 
The Setup is an offline stage, whereas the other three stages happen online.
An offline stage does not depend on the user's query and the server can perform this stage before the user sends its query to reduce latency. \Cref{tab:steps} summarizes the stages of our PIR protocol.
In the following sections, we describe each stage in detail.

\begin{table}
    \centering
    \caption[Stages of PIR using constant-weight codewords]{Stages of PIR using constant-weight codewords}
    \label{tab:steps}
    \resizebox{\columnwidth}{!}{%
    \begin{tabular}{c|c|c|c} \toprule
         \textbf{Stage} & \textbf{Performed by} & \textbf{Functionality} & 
        \begin{tabular}[c]{@{}c@{}}\textbf{Comp.}\\\textbf{Complexity}\end{tabular}
         \\ \midrule
         Setup & Server (Offline) & Set Parameters, Put DB in plaintext format & $O(n)$\\ \cline{1-4}
         Query & Client & Construct query, Send to Server & $O(m)$\\ \cline{1-4}
         \multirow{3}{*}{Process} & \multirow{3}{*}{Server} & Query Expansion & $O(m)$\\ \cline{3-4}
          &  & Selection Vector Calculation & $O(n)$ \\ \cline{3-4}
          &  & Inner product with DB & $O(ns)$\\ \cline{1-4}
         Extract & Client & Decrypt \& decode the server's response & $O(s)$ \\ \bottomrule
    \end{tabular}
    }
\end{table}

\subsubsection{Setup}

In this stage, parameters for the homomorphic encryption system are chosen such that they meet the security requirements. The payload data within each row of the database is then converted into FV plaintexts. Only the contents of each database row must be converted to plaintexts, not the set of identifiers.
However, the constant-weight code corresponding to each identifier can be calculated and stored in this stage to reduce the runtime in the online stages.
This stage can be done without regard to the user's query and only depends on the choice of encryption parameters. After this offline stage, the server holds a table of plaintexts with $n$ rows and at most $s$ plaintexts in each row, for some $s \geq 1$.

\subsubsection{Query}
In this stage, the user constructs its query in the appropriate format and sends it to the server. First, parameters for the user's query are chosen.
The Hamming weight, $k$, is chosen and then the code length $m$, is derived such that $\binom{m}{k} \geq n$.
The user then constructs its query as depicted by \Cref{alg:plaintext-query}. Let $q\in S(\mathbb{ID})$ denote the user's query. The user maps its query to a constant-weight codeword from $CW(m,k)$.
Let $E_q$ denote the mapping of $q$. $E_q$ is then converted to FV plaintexts as shown in lines 2--4 of  \Cref{alg:plaintext-query}.
The compression factor, $c$, indicates how many bits of the user's query are in each plaintext. Specifically, for $c\in \{0,1,...,\log_2 N\}$, exactly $2^c$ bits are in each plaintext. A higher compression factor reduces the upload cost but requires more computation for decompression, as we will see the next stage.
Finally, the plaintexts are encrypted using the user's secret key.
The client sends the output of \Cref{alg:plaintext-query} along with $m$, $k$, and $c$ to the server for the next stage.

\begin{algorithm}[!ht]
	 \caption[]{\textsc{Query}}
	 \label{alg:plaintext-query}
	 	\vspace{-3mm}
	 \begin{flushleft}
		 \textbf{Input:} $q\in S(\mathbb{ID})$, $m, k \in \NN$, $c \in \{0, 1, ..., \log_2 N\}$
	 \end{flushleft}
	 	\vspace{-3mm}
	 \begin{algorithmic}[1]
	    \State $E_q \leftarrow \texttt{MapToConstantWeightCode}(q, m, k)$
	    \State $h = \lceil\frac{m}{2^c} \rceil$
	 	\For {$i \in [h]$}
	 		\State $m_i(x) = \sum\limits_{j \in [2^c]} 2^{-c} \cdot E_q[i2^{c} + j] \cdot x^j$
	 	\EndFor
	    \For {$i \in [h]$}
	 		\State $ct_i(x) = \texttt{Enc}(\texttt{sk}, m_i(x))$
	 	\EndFor
	 \end{algorithmic}
	 	\vspace{-3mm}
	 \begin{flushleft}
    	 \textbf{Output:}  $(ct_i(x))_{i\in [h]}$
	 \end{flushleft}
	 	\vspace{-3mm}
\end{algorithm}

\subsubsection{Process Query}
This stage consists of three steps which are done by the server: Query Expansion, Selection Vector Calculation, and Inner Product.

\paragraph{Query Expansion.} In the first step, the server expands the ciphertexts received from the user such that each bit of the user's query is in a separate ciphertext.
\Cref{alg:oblivious-expand} describes the query expansion procedure, which is a modified version of \Cref{alg:oblivious-expand-sealpir}. We replace the use of two substitutions and one plaintext multiplication in the inner loop of \Cref{alg:oblivious-expand-sealpir} with one substitution and two plaintext multiplications.
Since substitution is slower compared to plain multiplication, as indicated in \Cref{tab:seal-ops}, there is an overall speedup. This modification in the expansion algorithm was first adopted in the implementation of MulPIR from the OpenMined community.\footnote{\url{https://github.com/OpenMined/PIR}}

\begin{algorithm}[!ht]
	 \caption[]{\textsc{Query Expansion}}
	 \label{alg:oblivious-expand}
	 	\vspace{-3mm}
	 \begin{flushleft}
		 \textbf{Input:} $(ct_j(x)) \in \mathcal{C}^{\ceil{\frac{m}{2^c}}}$, $m\in\NN$, $c \in \{0,1,...,\log_2 N\}$
	 \end{flushleft}
	 	\vspace{-3mm}
	 \begin{algorithmic}[1]
	    \State $h = \ceil{\frac{m}{2^c}}$
	    \State  $ctxts \leftarrow []$
	 	\For {$j\in[h]$}
	 		\State $cts \leftarrow [ct_j]$
	 	    \For {$a\in[c]$}
	 	        \For {$b\in[2^a]$}
	 	            \State $c_0 \leftarrow cts[b]$
	 	            \State $c_0 \leftarrow \texttt{Sub}_{N/2^{a}+1}(c_0)$
                    \State $c_1 \leftarrow x^{-2^a} \cdot c_0$

	 	            \State $cts[b + 2^a] \leftarrow x^{-2^a} \cdot cts[b]$
                    \State $cts[b] \leftarrow cts[b] + c_0$
                    \State $cts[b + 2^a] \leftarrow cts[b + 2^a] - c_1$
	 	        \EndFor
	 	    \EndFor
	 	    \State $ctxts \leftarrow ctxts || cts$
	 	\EndFor
	 \end{algorithmic}
	 	\vspace{-3mm}
	 \begin{flushleft}
    	 \textbf{Output:}  $ctxts \in C^m$\\
	 \end{flushleft}
	 	\vspace{-3mm}
\end{algorithm}

In \Cref{appendix-expansion}, we prove the correctness of this procedure by showing it is equivalent to \Cref{alg:oblivious-expand-sealpir}, which has been proven to be correct by Angel et al.~\cite{sealpir}. The for loop on line 6 of \Cref{alg:oblivious-expand} can be executed in parallel.

The output of this step is a vector of $m$ ciphertexts, where each ciphertext contains one of the bits of $E_q$, i.e., the encoded query. 

\paragraph{Selection Vector Calculation.} In this step, the server creates the selection vector using the expanded query from the output of the previous step.
For this, the server iterates over $\mathbb{ID}$, the set of identifiers in the database, maps each identifier to a constant-weight codeword and performs the equality operator between the mapped identifier and the user's query.
The constant-weight codeword corresponding to each identifier is calculated in the Setup stage to reduce online runtime.
We use the plain constant-weight equality operator since one of the operators is unencrypted.
\Cref{alg:selection-vector} depicts this step with the output from the query expansion as input.

\begin{algorithm}[!ht]
	 \caption[]{\textsc{Selection Vector Calculation}}
	 \label{alg:selection-vector}
	 	\vspace{-3mm}
	 \begin{flushleft}
		 \textbf{Input:} $ctxts \in \mathcal{C}^{m}$\\
	 \end{flushleft}
	 	\vspace{-3mm}
	 \begin{algorithmic}[1]
        \State $sel \leftarrow []$
	    \For {$i \in [n]$}
	        \State $E \leftarrow \texttt{MapToConstantWeightCode}(\mathbb{ID}[i], m, k)$
	        \State $sel[i] = \prod\limits_{E[j]=1} ctxts[j]$
	    \EndFor
	 \end{algorithmic}
	 	\vspace{-3mm}
	 \begin{flushleft}
    	 \textbf{Output:}  $sel \in \mathcal{C}^{n}$
	 \end{flushleft}
	 	\vspace{-3mm}
\end{algorithm}

This is the most computationally expensive step of the protocol, however, it can be done in parallel across the identifiers in the database. The output of this stage is an encrypted selection vector of size $n$, with each bit in a separate ciphertext.

\paragraph{Inner Product.}
In the last step of this stage, an inner product is performed between the selection vector derived from the previous step and the database. Each row of the database contains at most $s$ plaintexts from the setup phase, hence $s$ inner products are performed and $s$ ciphertexts are sent to the user as the response. Each inner product operation includes $n$ plaintext multiplication which can be done in parallel. The $s$ inner products can also be done in parallel when $s$ is large to enhance performance. The output of the inner products is sent to the user for the next stage.

\subsubsection{Extract}
In the last stage, the user decrypts the ciphertext(s) received from the server. The results are extracted from the decrypted messages by the client.

\subsubsection{Constant-weight Keyword PIR}
\label{sec:pir-sparse}

Recall that $\mathbb{ID}$ is the list of identifiers in the database, and $S(\mathbb{ID})$ refers to the domain of identifiers, i.e., the set of all possible identifiers. By definition, $\mathbb{ID} \subseteq S(\mathbb{ID})$. In the previous sections, we have discussed PIR in the case where $\mathbb{ID}=S(\mathbb{ID})$. Related work has also mainly focused on PIR under this assumption~\cite{sealpir,Ali2019CommunicationComputationTI}. A sparse database, however, specifies the case where $\mathbb{ID}$ where is much smaller than $S(\mathbb{ID})$. In this case, not all identifiers in the domain are associated with an element in the database.

The architecture described in this section is applicable when the database is sparse, with computation on the order of $|\mathbb{ID}|$, not $|S(\mathbb{ID})|$. For this, the following changes must be made to the protocol.

\begin{itemize}
\itemsep0mm
    \item In the query stage, the code length, $m$, and Hamming weight, $k$, are chosen such that $\binom{m}{k} \geq |S(\mathbb{ID})|$.
    \item In the selection vector calculation step, encrypted bits of the selection vector are generated only for identifiers in the database, i.e., the for loop on line 4 of \Cref{alg:selection-vector} is performed only over the identifiers in the database. Hence, this step is unchanged.
    \item Similarly in the inner product step, we only perform plain multiplications and sum for identifiers in the database.
\end{itemize}

PIR solutions based on selection vectors have a computational complexity that depends on the domain size, which makes them unsuitable for keyword PIR. We examine this further \Cref{sec:constant-weight-keyword-pir}.

%% file: evaluation.tex
\section{Evaluation of Equality Operators}
We evaluate equality operators in two categories:
\begin{itemize}
    \item Plain equality operators, where one operand is public, i.e., the circuit depends on one of the operands. We consider two candidates in this category: the plain folklore and the plain constant-weight equality operator.
    \item Arithmetic equality operators, where the circuit is oblivious to both operands and operates over an arbitrary field. We consider the arithmetic folklore and the arithmetic constant-weight equality operators in this category.
\end{itemize}

\Cref{tab:eq-ops-properties} summarizes these operators, along with the properties of circuits that implement each of them. We include properties that significantly influence the runtime such as the number of homomorphic and plain multiplications and the multiplicative depth. Note that different circuits operate over different domains, which are stated in \Cref{tab:eq-ops-properties}, but for a fair comparison, we select parameters such that the size of all the domains is at least $n$. To meet this criteria, the required condition for each of the operators is listed in the table.

\begin{table}[!ht]
\centering
\caption{Properties of circuits implementing equality operators mentioned in this work.
}
\label{tab:eq-ops-properties}
\resizebox{\columnwidth}{!}{
\begin{tabular}{ccccc}
\toprule
    \textbf{Operator} & \textbf{Domain} & \textbf{\# of Operations} & \begin{tabular}{c} \textbf{Multiplicative}\\\textbf{Depth} \end{tabular}  & \textbf{Conditions} \\ \midrule
    Plain Fl. & $\binaryset$ & $\ell \cdot \texttt{M}$ & $\ceil{\log_2 \ell}$ & $\ell \geq \log_2 n$ \\
    Plain Cw & $CW(m,k)$ & $k \cdot \texttt{M}$  & $\ceil{\log_2 k}$ & $\binom{m}{k} \geq n $ \\ \midrule
    Arithmetic Fl. & $\binaryset$ & $ 2\ell \cdot \texttt{M}$ & $1+\ceil{\log_2 \ell}$ & $\ell \geq \log_2 n$ \\
    Arithmetic Cw & $CW(m,k)$ & $\texttt{PM} + (m+k) \cdot \texttt{M}$  & $\ceil{\log_2 k}$ & $\binom{m}{k} \geq n $ \\
    \bottomrule
\end{tabular}
}
\end{table}

In the experiments, we vary the domain size, $n$, to observe the effect on the performance of the circuit implementing each operator.
Our implementation of all the equality circuits is open-source and available on Github\footnote{\url{https://github.com/RasoulAM/constant-weight-pir}}.
We implement the circuits using C++ and the SEAL library (version 3.6).
For the SEAL library, we use three different encryption parameters specified by $N$, the polynomial modulus degree, where $N\in\{4096, 8192, 16384\}$. The default ciphertext modulus is used to achieve 128-bit security. We also run all experiments both in single-thread and in parallel across multiple cores. The goal is to observe the speedup in each circuit when run in parallel.

All circuits are run in a SIMD fashion using the batch encoding functionality of SEAL. Using this feature, $N$ elements can be compared at the same time. In plain operators, since the circuit depends on the plain operands, $N$ elements are compared to the same operand in the clear. This is not the case for the arithmetic operand, in which $N$ pairs of numbers are compared simultaneously. The runtime can be divided by $N$ to achieve the amortized cost of one equality check.

We run all experiments on an Intel Xeon E5-4640 @ 2.40GHz server running Ubuntu 16.04. Parallelization is performed using 32 physical cores.

\subsection{Plain Operators}

\Cref{tab:plain-eq} summarizes the results of our experiments for plain equality operators. Each column reports the runtimes for a specific domain size. We report the results for the plain constant-weight operator in four categories based on the relationship between $\log_2 n$ and $k$.

\input{tables/plain-eq}

The constant-weight plain operator consistently outperforms the folklore operator in terms of running time. The advantage is greater when smaller homomorphic encryption parameters (namely $N$) can be used. This is possible due to a smaller multiplicative depth compared to the folklore circuit in cases where $k < \log_2 n$. However, the advantage exists even when using the same homomorphic encryption parameters. This can be attributed to fewer multiplications in the circuit ($k$ compared to $\ell$) when a small Hamming weight is used.

Faster runtimes for the plain constant-weight circuit come at the cost of higher memory usage during the protocol. The memory usage depends on the code length, also specified in the table. Depending on the application, the code length determines the communication complexity if operands are communicated over the network.

Parallelization offers roughly up to 10$\times$ speedup for both circuits and there is no noticeable difference in the advantage that parallel implementation offers for both circuits. \Cref{tab:eq-parallel} in \Cref{sec:parallel-plain} shows the runtimes of plain operators when parallelized.

\subsection{Arithmetic Operators}
\Cref{tab:arith-eq} summarizes the results of our experiments for arithmetic equality operators. Similar to before, each column reports the runtimes for a specific domain size. We report the results for the arithmetic constant-weight operator in four categories based on the relationship between $\log_2 n$ and $k$.

\input{tables/arith-eq}

Unlike the plain operators, the constant-weight arithmetic operator is not always faster than the equivalent folklore arithmetic equality circuit, or the advantage is marginal. This is due to the large number of homomorphic multiplications that are required in a constant-weight arithmetic circuit ($m+k$ compared to $\ell$). Specifically, when the constant-weight code length, $m$, is large due to a small Hamming weight, $k$, the number of multiplications can be very high compared to the folklore.
However, in some cases, the smaller Hamming weight results in a lower multiplicative depth, which in turn allows the use of smaller homomorphic encryption parameters. For example for $n=2^{16}$, the constant-weight operator with $k=4$ using $N=8192$ is about 4 times faster than the folklore using $N=16384$. The amortized cost is also about 2 times faster.

Similar to the plain equality operators, high memory usage is also an issue with the arithmetic constant-weight equality operator and it requires much more memory than the equivalent folklore operator.

The effect of the parallelization is however substantially different between folklore and constant-weight operators.
\Cref{fig:speedup-arith-eq} shows the speedup for each of the five categories in \Cref{tab:arith-eq}. The folklore circuit runs at most 2 times faster with parallelization, whereas the constant-weight circuit has more than a 10$\times$ speedup in some cases. The speedup is larger as the domain size grows. The speedup is mainly due to the $m$ homomorphic multiplications that can be done in parallel. With parallelization, the arithmetic constant-weight operator outperforms the arithmetic folklore operators for all domain sizes.
\Cref{tab:eq-parallel} in \Cref{sec:parallel-plain} shows the runtimes of arithmetic operators when run in parallel.

\begin{figure}[!ht]
\centering
\resizebox{\columnwidth}{!}{
    \begin{tikzpicture}
    \begin{axis}[
        width=\textwidth,
        height=0.45\textwidth,
        xlabel={\Large Domain Bit-length},
        ylabel={\Large Speedup},
        xtick={16,64,128,256,512},
        table/col sep=comma,
        legend style={at={(0.95,0.6)}}
      ]
        \addplot[color=blue, mark=square*] table [y expr=\thisrow{fl}/\thisrow{fl-p} , x=domain-bit-length]{data/arith-parallel.csv};
        \addlegendentry{Folklore}

        \addplot[color=red, mark=square*] table [y expr=\thisrow{cw1}/\thisrow{cw1-p} , x=domain-bit-length]{data/arith-parallel.csv};
        \addlegendentry{Constant-weight ($k=\log_2 n$)}
        
        \addplot[color=black!60!green, mark=square*] table [y expr=\thisrow{cw2}/\thisrow{cw2-p} , x=domain-bit-length]{data/arith-parallel.csv};
        \addlegendentry{Constant-weight ($k=\frac{1}{2}\log_2 n$)}
        
        \addplot[color=purple, mark=square*] table [y expr=\thisrow{cw4}/\thisrow{cw4-p} , x=domain-bit-length]{data/arith-parallel.csv};
        \addlegendentry{Constant-weight ($k=\frac{1}{4}\log_2 n$)}

        \addplot[color=orange, mark=square*] table [y expr=\thisrow{cw8}/\thisrow{cw8-p} , x=domain-bit-length]{data/arith-parallel.csv};
        \addlegendentry{Constant-weight ($k=\frac{1}{8}\log_2 n$)}
        
    \end{axis}
    \end{tikzpicture}
    }
    \caption{Speedup using parallelization when evaluating arithmetic equality operators}
    \label{fig:speedup-arith-eq}
\end{figure}
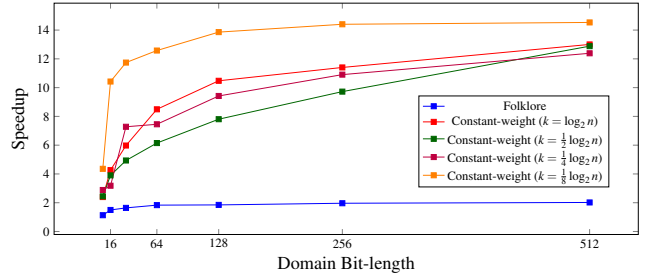

\section{Evaluation of PIR for Large Payloads}

In this section, we evaluate PIR protocols based on runtime and communication cost.
Specifically, we compare PIR using the folklore equality operator (which we call folklore PIR), Constant-weight PIR, SealPIR~\cite{Ali2019CommunicationComputationTI}, and  MulPIR~\cite{Ali2019CommunicationComputationTI}.

Folklore PIR refers to a PIR protocol using the same architecture as constant-weight PIR, but replacing the equality operator with the plain folklore operator. Indices are encoded using the logarithmic binary encoding in this protocol.

SealPIR and MulPIR are based on the approach where the selection vector is communicated to the server, whereas folklore PIR and constant-weight PIR make use of equality operators. We aim to compare the two general methods (selection vectors vs.\ equality circuits) while also evaluating constant-weight PIR against folklore PIR.

\paragraph{Unary Approach.} Note that SealPIR and MulPIR with $d=1$ are equivalent to constant-weight PIR when $k=1$. Hence, we refer to this configuration as the \emph{unary} approach. We report the runtimes of this approach in \Cref{sec:other-approach} as a baseline. To give a summary, the unary approach has a smaller runtime compared to the other approaches described in this paper and has a reasonable upload cost for small, packed databases. However, the upload cost is on the order of the size of the domain which is impractical for large domains. Hence, we exclude it from the comparison in this section. Specifically, we compare approaches that have a multiplicative depth of at least one. This includes SealPIR and MulPIR with $d\geq 2$, and constant-weight PIR with $k\geq 2$. This setup is particularly useful for large domains, which we explain in \Cref{sec:pir-sparse}.

\paragraph{Implementation Details.} Constant-weight PIR is implemented as described in \Cref{pir-constant-weight-code}. We also implement folklore PIR using the same architecture and consisting of the same stages described in \Cref{pir-constant-weight-code}. However, we use a logarithmic binary encoding for indices and the equality operator is replaced with a plain folklore equality operator per definition in \Cref{folklore-plain-eq}.

Our implementation of constant-weight PIR and folklore PIR is open-source and available on Github\footnote{\url{https://github.com/RasoulAM/constant-weight-pir}}. We implement all protocols using C++ and SEAL (version 3.7)\footnote{\url{https://github.com/microsoft/SEAL}} as the homomorphic encryption library. For SealPIR and MulPIR, we use the implementation by the OpenMined community\footnote{\url{https://github.com/OpenMined/PIR}}.
We select homomorphic encryption parameters such that it satisfies 128-bit security. Specifically, we use $N\in\{4096,8192,16384\}$ and the default coefficient modulus in SEAL for 128-bit security. Each protocol is run with the smallest parameter set which produces decryptable results. Specifically, SealPIR uses $N=4096$, whereas MulPIR, folklore PIR, and constant-weight PIR require $N\geq8192$.

We run all experiments on an Intel Xeon E5-4640 @ 2.40GHz server running Ubuntu 16.04. 

\paragraph{Experimental Setup.}
Index PIR implies that all database rows are full (in contrast to keyword PIR where some keywords do not correspond to any payload data in the database).
We are interested in the case where the payload is large.
Previous work on PIR, specifically information theoretic PIR, has examined PIR when the payload grows arbitrarily large~\cite{8006859,7997393}. There also exist applications of single-server PIR such that the payload can be arbitrarily large~\cite{Mayberry2014EfficientPF}.

Note that the size of the payload data is a multiple of the plaintext size and plaintext sizes depend on the homomorphic encryption parameters used in each approach. Hence, we run experiments for a payload data of one plaintext and extrapolate the results for larger payload data sizes.

\paragraph{Results.} \Cref{tab:compare-pir-params} lists the properties of the four aforementioned protocols.

\begin{table}[!ht]
    \centering
    \caption{Parameters for PIR protocols when $|S(\mathbb{ID})|=|\mathbb{ID}|=n$ and the payload data is $s$ plaintexts.}
    \label{tab:compare-pir-params}
\resizebox{\columnwidth}{!}{
    \begin{tabular}{ccccc}
    \toprule
        \textbf{Method} & \textbf{Mult Depth} & \begin{tabular}[c]{@{}c@{}}\textbf{Query}\\\textbf{Bit-length}\end{tabular} & \begin{tabular}[c]{@{}c@{}} \textbf{\# of Operations}\\\textbf{(Excluding Expansion)} \end{tabular} & \begin{tabular}[c]{@{}c@{}}\textbf{Download}\\\textbf{Cost (in cts)}\end{tabular}\\ 
    \midrule
        SealPIR & $d-1$ & $d \ceil{\sqrt[d]{\dbsize}\ }$ & $(\sum_{i=0}^{d-1} \dbsize^{\frac{d-i}{d}}F^i \cdot \texttt{PM}) \cdot s$ & $F^{d-1}s$ \\ \midrule
        MulPIR & $d-1$ & $d \ceil{\sqrt[d]{\dbsize}\ }$ & $(\dbsize \cdot \texttt{PM} + \sum_{i=1}^{d-1} \dbsize^{\frac{d-i}{d}} \cdot \texttt{M} ) \cdot s$ & $s$\\ \midrule
        Fl. PIR & $\ceil{\log_2 \ceil{\log_2 \dbsize}}$ & $\ceil{\log_2 \dbsize}$ & $ \dbsize \ceil{\log_2 n} \cdot \texttt{M} + \dbsize s \cdot \texttt{PM} $ & $s$ \\ \midrule
        Cw PIR & $\ceil{\log k}$ & $O\left(\sqrt[k]{k!\dbsize} + k\right)$ & $ \dbsize k \cdot \texttt{M} + \dbsize s \cdot \texttt{PM} $ & $s$ \\
    \bottomrule
    \end{tabular}
}
\end{table}

First, we compare protocols using equality operators. \Cref{tab:packed-pir} compares folklore PIR and constant-weight PIR.
This table shows folklore PIR is much slower than constant-weight PIR. At $n=512$, the parameters of the homomorphic cryptosystem must be increased from $N=8192$ to $N=16384$ to produce valid, decryptable results. Larger parameters increase the runtime drastically. Consequently, constant-weight PIR is the first practical PIR protocol using equality operators. \Cref{tab:packed-pir} includes runtimes for constant-weight PIR when run in parallel to demonstrate practicality.

\input{tables/pir-packed}

Another observation from \Cref{tab:compare-pir-params} is that the download cost of SealPIR is larger compared to the other protocols.
\Cref{tab:comm-cost} shows the upload, download, and total communication cost for a payload data of one plaintext.
For larger payloads, the high download cost of SealPIR is multiplied by the number of plaintexts in the payload data. Hence, constant-weight PIR and MulPIR have a lower communication cost for large payload data and streaming data.

\begin{table}[!ht]
    \caption{Upload, download, and total communication cost for payload data equal to one plaintext.}
    \label{tab:comm-cost}
    \centering
    \resizebox{\columnwidth}{!}{
        \begin{tabular}{cccc}
            \toprule
             & \textbf{Upload Cost} & \textbf{Download Cost} & \textbf{Total Comm.} \\
             \midrule
            SealPIR & 61.4 KB & 307 KB & 368.4 KB\\
            MulPIR & 122 KB & 119 KB & 241 KB\\
            Constant-weight PIR & 216 KB & 106 KB & 322 KB\\
            \bottomrule
        \end{tabular}
    }
\end{table}

Next, we analyze the effect of larger payload data on the runtime of the protocols.
We focus our attention to comparing MulPIR and constant-weight PIR as they have similar communication complexity.
Runtimes for SealPIR are given in \Cref{sec:other-approach}.
MulPIR (and SealPIR) must repeat the server computation (except the expansion step) for each plaintext in the payload. This applies to other approaches using selection vectors as well.
In constant-weight PIR, only the inner product step must be repeated for each plaintext in the payload.

To show this effect, we perform PIR over a database with $n=16384$ rows with various, large, payload sizes.
\Cref{fig:pir-large-runtimes} shows the runtime of the constant-weight PIR (with $k=2$) as a function the payload size. The implementation of MulPIR by OpenMined does not support large payloads, so we provide a lower bound of the runtime of MulPIR (with $d=2$) based on the server time for a payload of one plaintext.

As seen in \Cref{fig:pir-large-runtimes}, the runtime of constant-weight PIR is higher than MulPIR for a small payload but grows at a slower rate. Eventually, constant-weight PIR outperforms MulPIR when the payload size exceeds 268 KB. This corresponds to a database size of about 4.3 GB.

\begin{figure}[!ht]
\centering
\resizebox{\columnwidth}{!}{
    \begin{tikzpicture}
    \begin{axis}[
        width=\textwidth,
        height=0.45\textwidth,
        xlabel=Payload Data (KB),
        ylabel=Runtime (in seconds),
        table/col sep=comma,
        legend style={at={(0.46,0.95)}}
      ]
        \addplot[color=blue, mark=square*] table [y=mulpir, x=response-size-kbs]{data/large-pir.csv};
        \addlegendentry{MulPIR, $d=2$ (lower bound estimate)}
        \addplot[color=red, mark=square*] table [y=constant-weight-single-thread, x=response-size-kbs]{data/large-pir.csv};
        \addlegendentry{Constant-weight, $k=2$}
        \addplot[color=orange, mark=square*] table [y=constant-weight parallel, x=response-size-kbs]{data/large-pir.csv};
        \addlegendentry{Constant-weight, $k=2$ (parallelized)}
    \end{axis}
    \end{tikzpicture}
    }
    \caption{Runtime of constant-weight PIR and an estimation of the runtime of MulPIR for large payloads.}
    \label{fig:pir-large-runtimes}
\end{figure}
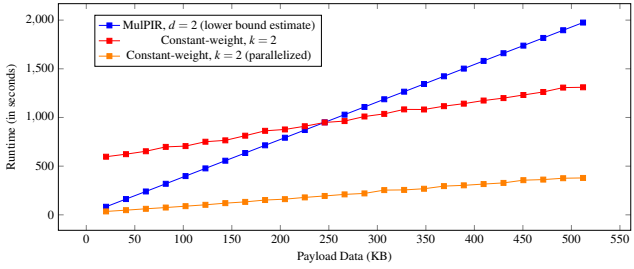

To summarize, constant-weight PIR outperforms folklore PIR in all cases. It also has a smaller communication complexity and lower runtime compared to SealPIR and MulPIR, respectively, when the payload size increases.

\section{Analysis of Constant-weight Keyword PIR}
\label{sec:constant-weight-keyword-pir}

In this section, we show how constant-weight PIR performs over a sparse database.
We first motivate our approach to keyword PIR by discussing the private file retrieval as an application. 
We also argue about the modifications required for MulPIR and SealPIR to allow for keyword PIR without reducing the problem to index PIR.

\paragraph{Keyword PIR for Private File Retrieval.}

Private file retrieval is a setup, similar to that of PIR, where the items that are retrieved are large, e.g. files or documents. Private retrieval of large items has been discussed in the literature~\cite{chor1995private} which differs from the case where only a single bit is retrieved. Solutions based on ORAM come at a high communication cost~\cite{Mayberry2014EfficientPF, Shi2011ObliviousRW}. Specifically, the response size is $O(\ell \log n)$ in the worst case for retrieving an item of size $\ell$ amongst $n$ elements.

PIR is a suitable solution for this problem given that it can achieve asymptotically optimal communication complexity~\cite{Mayberry2014EfficientPF}. Keyword PIR provides the additional feature of retrieving documents by identifiers instead of an index in a directory.

Constant-weight keyword PIR is a practical keyword PIR protocol that can be used for private file retrieval. Moreover, constant-weight PIR is performed without the use of a hash-table to store the identifiers or multiple rounds of communication, which is in contrast to the existing approaches for keyword PIR~\cite{chor1997private, Ali2019CommunicationComputationTI}. This is useful in the presence of many users with unreliable connections and low bandwidth. Particularly in solutions that store the identifiers using a hash-table, updates to the database may require a change in the parameters of the hash function to avoid collisions. An additional round of communication is required for each query to communicate new hash function parameters to the user.

\Cref{tab:pir-apps} shows example runtimes of constant-weight PIR used to retrieve files from a database with large items. In the next subsection, we provide a finer analysis of the cost of constant-weight keyword PIR.
The experiments in \Cref{tab:pir-apps} were performed on an Intel(R) Xeon(R) CPU E7-8860 v4 @ 2.20GHz running Ubuntu 20.04. The experiments are parallelized over 144 cores to achieve the best possible performance. The results are only to demonstrate practicality and can easily be enhanced using hardware accelerators (GPUs) or accelerators for the homomorphic encryption libraries such as HEXL~\cite{boemer2021intel}.

\begin{table}[!ht]
    \caption{Server runtimes for Constant-weight PIR of large payloads.}
    \label{tab:pir-apps}
    \centering
\resizebox{\columnwidth}{!}{
    \begin{tabular}{ccccc}
    \toprule
        \begin{tabular}{c} \textbf{Keyword}\\\textbf{Bitlength} \end{tabular} &
        \begin{tabular}{c} \textbf{Number of}\\\textbf{Items ($n$)} \end{tabular} &
         \begin{tabular}{c}\textbf{Database}\\\textbf{Size (GB)}\end{tabular} &
         \begin{tabular}{c}\textbf{Item}\\\textbf{Size (MB)}\end{tabular} &
         \begin{tabular}{c}\textbf{Server}\\\textbf{Time (s)}\end{tabular} \\ 
    \midrule
        \multirow{8}{*}{16} & \multirow{4}{*}{1000} & 1.3 &	1.3 &	51.9 \\
        & & 2.6 &	2.6 &	107 \\
        & & 5.2 &	5.2 &	200 \\
        & & 10.0 &	10.0 &	369 \\ \cline{2-5}
        & \multirow{4}{*}{10000} & 13.0 &	1.3 &	508 \\
        & & 26.0 &	2.6 &	878 \\
        & & 52.0 &	5.2 &	1670 \\
        & & 100.0 &	10.0 &	3250 \\
    \midrule
        \multirow{8}{*}{32} & \multirow{4}{*}{1000} & 1.3 &	1.3 &	59 \\
        & & 2.6 &	2.6 &	111 \\
        & & 5.2 &	5.2 &	212 \\
        & & 10.0 &	10.0 &	354 \\
        \cline{2-5}
         & \multirow{4}{*}{10000} & 13.0 &	1.3 &	506 \\
         & & 26.0 &	2.6 &	869 \\
         & & 52.0 &	5.2 &	1700 \\
         & & 100.0 &	10.0 &	3180 \\
    \midrule
        \multirow{8}{*}{48} & \multirow{4}{*}{1000} &  1.3 & 1.3 &	71.3 \\
         & & 2.6 &	2.6 &	129 \\
         & & 5.2 &	5.2 &	208 \\
         & & 10.0 &	10.0 &	380 \\
         \cline{2-5}
         & \multirow{4}{*}{10000} & 13.0 &	1.3 &	541 \\
         & & 26.0 &	2.6 &	922 \\
         & & 52.0 &	5.2 &	1720 \\
         & & 100.0 &	10.0 &	3300 \\
         
         \bottomrule
    \end{tabular}
}
\end{table}

\paragraph{Analysis of PIR for Sparse Domains.}
\Cref{tab:compare-sparse-pir-params} shows the properties of the PIR protocols, adjusted for when the database is sparse. $n$ and $|S|$ denote the number of rows in the database and the size of the domain from which the query is selected, respectively.

\begin{table}[!ht]
    \caption{Properties of SealPIR, MulPIR, and constant-weight PIR when used for keyword PIR.}
    \label{tab:compare-sparse-pir-params}
    \centering
\resizebox{\columnwidth}{!}{
    \begin{tabular}{ccccc}
    \toprule
        \textbf{Method} & 
        \begin{tabular}{c} \textbf{Mult}\\\textbf{Depth} \end{tabular}
         & \begin{tabular}{c}\textbf{Query}\\\textbf{Bit-length}\end{tabular} & \begin{tabular}{c}\textbf{\# of Operations}\\\textbf{(Excluding Expansion)} \end{tabular} & \begin{tabular}[c]{@{}c@{}}\textbf{Download}\\\textbf{Cost (in cts)}\end{tabular}\\ 
    \midrule
        SealPIR & $d-1$ & $d \ceil{\sqrt[d]{|S|}\ }$ & $\dbsize \cdot \texttt{PM} + \sum_{i=1}^{d-1} {|S|}^{\frac{d-i}{d}}F^i \cdot \texttt{PM}$ & $F^{d-1}$ \\ \midrule
        MulPIR & $d-1$ & $d \ceil{\sqrt[d]{|S|}\ }$ & $\dbsize \cdot \texttt{PM} + \sum_{i=1}^{d-1} {|S|}^{\frac{d-i}{d}} \cdot \texttt{M} $ & $1$\\ \midrule
        CwPIR & $\ceil{\log k}$ & $O\left(\sqrt[k]{k!|S|} + k\right)$ & $ \dbsize k \cdot \texttt{M} + \dbsize \cdot \texttt{PM} $ & $1$ \\
    \bottomrule
    \end{tabular}
}
\end{table}

We argue that constant-weight PIR is minimally affected by sparsity in the database and it is a suitable solution for keyword PIR. \Cref{tab:compare-sparse-pir-params} supports this argument, as the number of operations (excluding expansion) for constant-weight PIR does not depend on the size of the domain. We exclude folklore PIR from this section entirely since it follows the same approach as constant-weight PIR and is strictly slower.

\Cref{tab:compare-sparse-pir-params} also shows the query bit-length of each method.
The query bit-length determines the communication cost in the protocol and also affects the computation cost, specifically the expansion step. 
This query bit-length is affected by the domain size and is equal to the length of the constant-weight code that is used in constant-weight PIR. SealPIR and MulPIR use the same type of encoding for PIR queries which essentially calculates the position of the desired row of the database when restructured into a $d$-dimensional table. We denote this as a \emph{dimension-wise} encoding in this section.

\input{tables/encoding-size}

\Cref{tab:encoding-size} shows the number of bits required to represent a query using a constant-weight codeword and a dimension-wise encoding as a function the domain bit-length, $\log_2 |S|$.
The constant-weight code length is shown for four different values of $k$, the Hamming weight.
In the last three columns, we derive the bit-length of the dimension-wise encoding.
The depth refers to the multiplicative depth in a PIR protocol using the set of parameters in that column.

There are multiple observations from this table. Firstly, larger $k$ or $d$ (and higher multiplicative depth in turn) drastically reduces the bit-length of the query. Given this observation, a fair comparison between the constant-weight code and dimension-wise encoding is comparing those with the same multiplicative depth since the multiplicative depth directly impacts the performance. For the same multiplicative depth, the constant-weight code is smaller than the dimension-wise encoding. \Cref{fig:encoding-size} visualizes this for even larger domain sizes and higher multiplicative depths. Note that the scale on the vertical axis is logarithmic and the gap between the size of the codes increases as the domain size increases and a larger multiplicative depth is used.

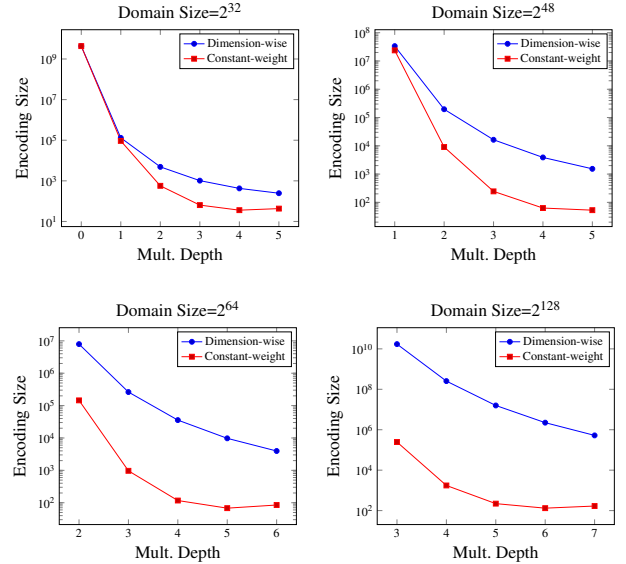
\begin{figure}[!ht]
\centering
    \subfloat{
        \begin{tikzpicture}[scale=0.46]
        \begin{axis}[
          xlabel={\Large Mult. Depth},
          ylabel={\Large Encoding Size},
          ymode=log,
          table/col sep=comma,
          title={\Large Domain Size=$2^{32}$}
          ]
            \addplot table [y=dimension-wise, x=depth]{data/encoding-size-len=32.csv};
            \addlegendentry{Dimension-wise}
            \addplot table [y=constant-weight, x=depth]{data/encoding-size-len=32.csv};
            \addlegendentry{Constant-weight}
        \end{axis}
        \end{tikzpicture}
    }~
    \subfloat{
        \begin{tikzpicture}[scale=0.46]
        \begin{axis}[
          xlabel={\Large Mult. Depth},
          ylabel={\Large Encoding Size},
          ymode=log,
          table/col sep=comma,
          title={\Large Domain Size=$2^{48}$}
          ]
            \addplot table [y=dimension-wise, x=depth]{data/encoding-size-len=48.csv};
            \addlegendentry{Dimension-wise}
            \addplot table [y=constant-weight, x=depth]{data/encoding-size-len=48.csv};
            \addlegendentry{Constant-weight}
        \end{axis}
        \end{tikzpicture}
    }
    
    \subfloat{
        \begin{tikzpicture}[scale=0.46]
        \begin{axis}[
          xlabel={\Large Mult. Depth},
          ylabel={\Large Encoding Size},
          ymode=log,
          table/col sep=comma,
          title={\Large Domain Size=$2^{64}$}
          ]
            \addplot table [y=dimension-wise, x=depth]{data/encoding-size-len=64.csv};
            \addlegendentry{Dimension-wise}
            \addplot table [y=constant-weight, x=depth]{data/encoding-size-len=64.csv};
            \addlegendentry{Constant-weight}
        \end{axis}
        \end{tikzpicture}
    }~
    \subfloat{
        \begin{tikzpicture}[scale=0.46]
        \begin{axis}[
          xlabel={\Large Mult. Depth},
          ylabel={\Large Encoding Size},
          ymode=log,
          table/col sep=comma,
          title={\Large Domain Size=$2^{128}$}
          ]
            \addplot table [y=dimension-wise, x=depth]{data/encoding-size-len=128.csv};
            \addlegendentry{Dimension-wise}
            \addplot table [y=constant-weight, x=depth]{data/encoding-size-len=128.csv};
            \addlegendentry{Constant-weight}
        \end{axis}
        \end{tikzpicture}
    }
    \caption{Encoding size as a function of multiplicative depth}
    \label{fig:encoding-size}
\end{figure}

The size of the query can also affect the server runtime in the protocol.
\Cref{fig:sparse-pir-runtime} shows the runtime of keyword PIR over a database of with $n=16384$ rows and payload size of one plaintext (roughly 20.1 KB) which corresponds to a database of about 330 MB.
We vary the domain size to examine the effect on the overall runtime, which is influenced by the query bit-length.

\begin{figure}[!ht]
\centering
    \begin{tikzpicture}
        \begin{axis}[
            width=\columnwidth,
            height=0.7\columnwidth,
            xlabel={\small Domain Bit-length ($\log_2 |S|$)},
            ylabel={\small Runtime (s)},
            table/col sep=comma,
            legend style={at={(0.30,0.95)}},
            ymin=0
          ]
            \addplot [name path = total2, color=blue, mark=*, mark size=0.8pt] table [y=time_server_total_k_2, x=log_domain_size]{data/sparse-pir.csv};
            \addlegendentry{$k=2$}
            \addplot [name path = total3, color={black!60!green}, mark=*, mark size=0.8pt] table [y=time_server_total_k_3, x=log_domain_size]{data/sparse-pir.csv};
            \addlegendentry{$k=3$}
            \addplot [name path = total4, color=red, mark=*, mark size=0.8pt] table [y=time_server_total_k_4, x=log_domain_size]{data/sparse-pir.csv};
            \addlegendentry{$k=4$}        

            \addplot [name path = iter2, color=blue!10, mark=none] table [y=time_process_k_2, x=log_domain_size]{data/sparse-pir.csv};
            \addplot [name path = iter3, color={green!10}, mark=none] table [y=time_process_k_3, x=log_domain_size]{data/sparse-pir.csv};
            \addplot [name path = iter4, color=red!10, mark=none] table [y=time_process_k_4, x=log_domain_size]{data/sparse-pir.csv};

            \addplot [blue!10] fill between[of=total2 and iter2];
            \addplot [green!10] fill between[of=total3 and iter3];
            \addplot [red!10] fill between[of=total4 and iter4];

        \end{axis}
    \end{tikzpicture}
\caption{Total server time of constant-weight keyword PIR as a function of the domain size for three different Hamming weights. The shaded areas indicate the amount of time required for the expansion step.}
\label{fig:sparse-pir-runtime}
\end{figure}
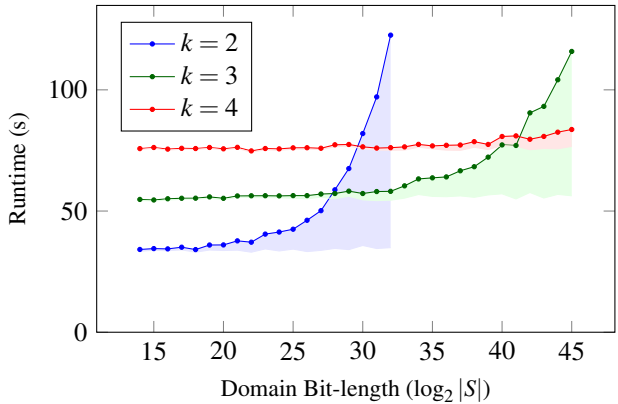

The runtime of the protocol consists of the expansion step, and the iteration step (which is the selection vector calculation and inner product combined). We report numbers for $k\in \{2,3,4\}$ since we know that $k=1$ produces an encoding size that is prohibitively large.
Each plot in \Cref{fig:sparse-pir-runtime} is for one value of $k$. The shaded area beneath each plot indicates the amount of time required for the expansion step.

Initially, for $\log_2 |S| \leq 27$, $k=2$ has the smallest server time. However, when $\log_2 |S|$ approaches $28$, the expansion time constitutes a significant portion of the server time and a switch to $k=3$ results in a smaller total server time.
Similarly, when $\log_2 |S|$ reaches 41, a switch to $k=4$ produces the best results.
Notice how the runtime excluding the expansion step does not change significantly for all values of $k$ and the time required for the expansion step eventually becomes the dominant factor when the domain size increases.

%% file: tables/plain-eq.tex
\begin{table}[!ht]
    \centering
    \caption[Runtimes for plain equality operators]{Runtimes for plain equality operators in seconds. Dashes indicate cases where the ciphertext was undecryptable due to homomorphic noise. $k$ and $m$ denote the Hamming weight and constant-weight code length, respectively. Bold numbers indicate the best runtimes for each $n$.}
    \label{tab:plain-eq}
    \resizebox{\columnwidth}{!}{
        \begin{tabular}{c|cccccccc} \toprule
         & $n$ & $2^{8}$& $2^{16}$ & $2^{32}$ & $2^{64}$ & $2^{128}$ & $2^{256}$ & $2^{512}$ \\ \bottomrule \toprule
        \multirow{4}{*}{\begin{tabular}[c]{@{}c@{}}Plain\\Folklore\end{tabular}} & $\ell$ & 8& 16 & 32 & 64& 128 & 256 & 512 \\ 
         & Mult Depth & 3& 4& 5& 6 & 7 & 8 & 9 \\     
         & $N=8192$ & 0.27 & 0.54 & - & - & - & - & - \\
         & $N=16384$& 1.1 & 2.4 & 5.0 & 10 & 21 & 42 & 84 \\ \bottomrule \toprule
        \multirow{5}{*}{\begin{tabular}[c]{@{}c@{}}Plain\\Constant-\\weight\\$k=\log_2 n$\end{tabular}} & \begin{tabular}[c]{@{}c@{}}$k$\end{tabular} &  8& 16 & 32 & 64& 128 & 256 & 512\\
         & Mult Depth & 3& 4& 5& 6 & 7 & 8 & 9 \\
         & \begin{tabular}[c]{@{}c@{}}$m$\end{tabular}& 12 & 22 & 43 & 85& 168 & 334 & 665\\
         & $N=8192$ & 0.27 & 0.57 & - & - & - & - & - \\
         & $N=16384$& 1.1 & 2.6 & 4.9 & 10 & 20 & 40 & 81 \\ \bottomrule \toprule
        \multirow{5}{*}{\begin{tabular}[c]{@{}c@{}}Plain\\Constant-\\weight\\$k=\frac{1}{2}\log_2 n$\end{tabular}} & \begin{tabular}[c]{@{}c@{}}$k$\end{tabular} & 4& 8& 16 & 32& 64& 128 & 256\\
         & Mult Depth &  2& 3& 4& 5 & 6 & 7 & 8 \\
         & \begin{tabular}[c]{@{}c@{}}$m$\end{tabular}& 11 & 19 & 36 & 68& 132 & 261 & 517 \\
         & $N=8192$ & 0.11 & 0.27 & 0.55 & - & - & - & - \\
         & $N=16384$& 0.49 & 1.1 & 2.4 & 5.0 & 10 & 21 & 41\\ \bottomrule \toprule
        \multirow{6}{*}{\begin{tabular}[c]{@{}c@{}}Plain\\Constant-\\weight\\$k=\frac{1}{4}\log_2 n$\end{tabular}} & \begin{tabular}[c]{@{}c@{}}$k$\end{tabular} & 2& 4& 8& 16& 32& 64& 128 \\
         & Mult Depth & 1& 2& 3& 4 & 5 & 6 & 7 \\
         & \begin{tabular}[c]{@{}c@{}}$m$\end{tabular}& 24 & 37 & 64 & 117 & 221 & 427 & 838\\
         & $N=4096$ & 0.01 & - & - & - & - & - & - \\
         & $N=8192$ & 0.04 & 0.12 & 0.25 & 0.49 & - & - & -\\
         & $N=16384$& 0.17 & 0.48 & 1.1 & 2.4 & 5.0 & 10 & 21 \\ \bottomrule \toprule
        \multirow{6}{*}{\begin{tabular}[c]{@{}c@{}}Plain\\Constant-\\weight\\$k=\frac{1}{8}\log_2 n$\end{tabular}} & \begin{tabular}[c]{@{}c@{}}$k$\end{tabular} & 1& 2& 4& 8 & 16& 32& 64\\
         & Mult Depth & 0& 1& 2& 3 & 4 & 5 & 6 \\
         & \begin{tabular}[c]{@{}c@{}}$m$\end{tabular}& 256& 363& 569& 968 & 1749& 3290& 6349\\
         & $N=4096$ & \textbf{0.0001} & \textbf{0.008}& - & - & - & - & - \\
         & $N=8192$ & 0.0004 & 0.038 & \textbf{0.10}& \textbf{0.25} & \textbf{0.54} & - & - \\
         & $N=16384$& 0.002 & 0.17 & 0.5 & 1.1 & 2.4 & \textbf{5.0}& \textbf{10}\\ \bottomrule
        \end{tabular}
    }
\end{table}

%% file: tables/arith-eq.tex
\begin{table}[!ht]
\centering
\caption{Runtimes for arithmetic equality operators in seconds. Dashes indicate cases where the ciphertext was undecryptable due to homomorphic noise. $k$ and $m$ denote the Hamming weight and constant-weight code length, respectively. Bold numbers indicate the best runtimes for each $n$.}
\label{tab:arith-eq}
\resizebox{\columnwidth}{!}{
\begin{tabular}{c|cccccccc} \toprule
     & $n$ & $2^{8}$& $2^{16}$ & $2^{32}$ & $2^{64}$ & $2^{128}$ & $2^{256}$ & $2^{512}$ \\ \bottomrule \toprule
    \multirow{4}{*}{\begin{tabular}[c]{@{}c@{}}Arithmetic\\Folklore\end{tabular}} & $\ell$ & 8& 16 & 32 & 64& 128 & 256 & 512 \\ 
     & Mult Depth & 4 & 5 & 6 & 7 & 8 & 9 & 10\\
     & $N=8192$ & \textbf{0.49} & - & - & - & - & - & - \\
     & $N=16384$ & 2.2 & 4.6 & 9.2 & 19 & 37 & 74 & 149 \\ \bottomrule \toprule
    \multirow{5}{*}{\begin{tabular}[c]{@{}c@{}}Arithmetic\\Constant-\\weight\\$k=\log_2 n$\end{tabular}} & $k$ & 8 & 16 & 32 & 64 & 128 & 256 & 512 \\
     & Mult Depth & 4 & 5 & 6 & 7 & 8 & 9 & 10 \\
     & $m$ & 12 & 22 & 43 & 85 & 168 & 334 & 665\\
     & $N=8192$ & 0.692 & - & - & - & - & - & - \\
     & $N=16384$ & 3.0 & 6.0 & 12 & 23 & 47 & 93 & 186\\ \bottomrule \toprule
    \multirow{5}{*}{\begin{tabular}[c]{@{}c@{}}Arithmetic\\Constant-\\weight\\$k=\frac{1}{2}\log_2 n$\end{tabular}} & $k$ & 4 & 8 & 16 & 32 & 64 & 128 & 256\\
     & Mult Depth& 3 & 4 & 5 & 6 & 7 & 8 & 9 \\
     & $m$ & 11 & 19 & 36 & 68 & 132 & 261 & 517\\
     & $N=8192$ & 0.53 & - & - & - & - & - & - \\
     & $N=16384$ & 2.2 & 4.3 & \textbf{8.2} & \textbf{16} & \textbf{31} & \textbf{63} & \textbf{123}\\ \bottomrule \toprule
    \multirow{5}{*}{\begin{tabular}[c]{@{}c@{}}Arithmetic\\Constant-\\weight\\$k=\frac{1}{4}\log_2 n$\end{tabular}} & $k$ & 2 & 4 & 8 & 16 & 32 & 64 & 128\\
     & Mult Depth & 2 & 3 & 4 & 5 & 6 & 7 & 8\\
     & $m$ & 24 & 37 & 64 & 117 & 221 & 427 & 838 \\
     & $N=8192$ & 0.85 & \textbf{1.3} & - & - & - & - & - \\
     & $N=16384$ & 4.3 & 6.4 & 11 & 21 & 40 & 78 & 154\\ \bottomrule \toprule
    \multirow{6}{*}{\begin{tabular}[c]{@{}c@{}}Arithmetic\\Constant-\\weight\\$k=\frac{1}{8}\log_2 n$\end{tabular}} & $k$ & 1 & 2 & 4 & 8 & 16 & 32 & 64 \\
     & Mult Depth & 1 & 2 & 3 & 4 & 5 & 6 & 7 \\
     & $m$ & 256 & 363 & 569 & 968 & 1749 & 3290 & 6349 \\
     & $N=4096$ & 2.0 & - & - & - & - & - & -\\
     & $N=8192$ & 8.4 & 12 & 19 & - & - & - & - \\
     & $N=16384$ & 41 & 58 & 91 & 156 & 282 & 533 & 1064 \\  \bottomrule
\end{tabular}
}
\end{table}

%% file: tables/pir-packed.tex
\begin{table}[!ht]
    \centering
    \caption{Runtime of PIR protocols using equality operators for a response size of one plaintext. Runtimes are in seconds and an average of 10 runs. *This parameter set did not produce a decryptable result.}
    \label{tab:packed-pir}
    \resizebox{\columnwidth}{!}{%
    \begin{tabular}{@{\extracolsep{4pt}}ccccccc@{}}
    \toprule
    &&& \multicolumn{4}{c}{\textbf{Time (s)}}\\ \cline{4-7}
    \begin{tabular}[c]{@{}c@{}}\textbf{\# of}\\\textbf{Rows}\end{tabular} &
    \begin{tabular}[c]{@{}c@{}}\textbf{DB Size}\\\textbf{(MB)}\end{tabular} &
    \begin{tabular}[c]{@{}c@{}}\textbf{Code}\\\textbf{Length}\end{tabular} & \begin{tabular}[c]{@{}c@{}}\textbf{Expansion}\end{tabular} & \begin{tabular}[c]{@{}c@{}}\textbf{Sel. Vec.}\\\textbf{Calculation}\end{tabular} & \begin{tabular}[c]{@{}c@{}}\textbf{Inner}\\\textbf{Product}\end{tabular} & \begin{tabular}[c]{@{}c@{}}\textbf{Total Server}\end{tabular}\\
    \bottomrule
    \toprule
    \multicolumn{7}{c}{Folklore, $N=8192$ (Query = 216 KB, Response = 106 KB)} \\
    \midrule
    256 &	8 &	5 &	0.06 &	58 &	0.9 & 60 \\
    512$^{*}$ & 9 &	10 & 0.1 & 130 & 1.7 & 130 \\
    \bottomrule
    \toprule
    \multicolumn{7}{c}{Folklore, $N=16384$ (Query = 913 KB, Response = 224 KB)} \\
    \midrule
    512 &	21 &	9 &	0.8 &	650 &	7.4 &	660 \\
    1024 &	42 &	10 &	0.8 &	1500 &	14 &	1500 \\
    2048 &	84 &	11 &	0.8 &	3300 &	29 &	3300 \\
    4096 &	170 &	12 &	0.8 &	7200 &	56 &	7200 \\
    8192 &	340 &	13 &	0.8 &	16000 &	120 &	16000 \\
    16384 &	670 &	14 &	0.8 &	35000 &	250 &	35000 \\
    \bottomrule
    \toprule
    \multicolumn{7}{c}{Constant-weight $k=2, N=8192$, (Query = 216 KB, Response = 106 KB)} \\
    \midrule
    \multicolumn{7}{c}{Single-thread} \\
    \midrule
    256 & 5.2 &	24 &	0.3 &	8.3 &	0.9 &	9.7 \\
    512 & 10 &	33 &	0.5 &	17 &	1.7 &	19 \\
    1024 & 21 &	46 &	0.5 &	33 &	3.5 &	38 \\
    2048 & 42 &	65 &	1 &	67 &	6.9 &	75 \\
    4096 & 84 &	92 &	1 &	130 &	13 &	150 \\
    8192 & 170 &	129 &	2 &	270 &	27 &	300 \\
    16384 & 340 &	182 &	2 &	540 &	55 &	600 \\
    32768 & 670 &	257 &	5 &	1100 &	110 &	1200 \\
    65536 & 1300 &	363 &	5 &	2300 &	230 &	2500 \\

    \bottomrule
    \multicolumn{7}{c}{Parallelized} \\
    \midrule
    256 & 5.2 &	24 &	0.1 &	0.5 &	0.3 &	1.1 \\
    512 & 10 &	33 &	0.1 &	0.7 &	0.5 &	1.6 \\
    1024 & 21 &	46 &	0.2 &	1.4 &	1.2 &	2.9 \\
    2048 & 42 &	65 &	0.2 &	2.9 &	2.4 &	5.6 \\
    4096 & 84 &	92 &	0.3 &	5.7 &	4.6 &	11 \\
    8192 & 170 &	129 &	0.3 &	11 &	9.2 &	21 \\
    16384 & 340 &	182 &	0.4 &	22 &	18 &	41 \\
    32768 & 670 &	257 &	0.6 &	44 &	34 &	79 \\
    65536 & 1300 &	363 &	0.7 &	87 &	70 &	160 \\
    131072 & 2700 &	513 &	1.2 &	170 &	140 &	320 \\
    262144 & 5400 &	725 &	1.4 &	340 &	290 &	640 \\
    \bottomrule
    \end{tabular}%
    }
    \end{table}

%% file: tables/encoding-size.tex
\begin{table}[!ht]
    \caption[Bit-length of the query in different protocols]{Bit-length of the query in different protocols}
    \label{tab:encoding-size}
    \centering
    \resizebox{\columnwidth}{!}{
        \begin{tabular}{c|c|c|c|c|c|c|c}\toprule
        \multirow{3}{*}{\textbf{\begin{tabular}[c]{@{}c@{}}Domain\\ Bit-length\\($\log_2 |S|$)\end{tabular}}}  & \multicolumn{4}{c|}{\textbf{Constant-weight code size}}  & \multicolumn{3}{c}{\textbf{Dimension-wise}}\\
         & \multicolumn{1}{c|}{\textbf{depth=0}}  & \multicolumn{1}{c|}{\textbf{depth=1}}  & \multicolumn{2}{c|}{\textbf{depth=2}}  & \multicolumn{1}{c|}{\textbf{depth=0}}  & \multicolumn{1}{c|}{\textbf{depth=1}}  & \multicolumn{1}{c}{\textbf{depth=2}} \\
          & \multicolumn{1}{c|}{\textbf{k=1}} & \multicolumn{1}{c|}{\textbf{k=2}}  & \multicolumn{1}{c|}{\textbf{k=3}} & \multicolumn{1}{c|}{\textbf{k=4}}  & \multicolumn{1}{c|}{\textbf{d=1}}  & \multicolumn{1}{c|}{\textbf{d=2}}  & \multicolumn{1}{c}{\textbf{d=3}} \\ \midrule
        4   & 16  & 7 & 6  & 7 & 16  & 8 & 9\\
        6   & 64  & 12  & 9  & 8 & 64  & 16  & 12 \\
        8   & 256 & 24  & 13 & 11  & 256 & 32  & 21 \\
        10  & 1024  & 46  & 20 & 15  & 1024  & 64  & 33 \\
        12  & 4096  & 92  & 31 & 20  & 4096  & 128 & 48 \\
        14  & 16384 & 182 & 48 & 27  & 16384 & 256 & 78 \\
        16  & 65536 & 363 & 75 & 37  & 65536 & 512 & 123\\
        18  & 262144  & 725 & 118  & 52  & 262144  & 1024  & 192\\
        20  & - & 1449  & 186  & 73  & - & 2048  & 306\\
        22  & - & 2897  & 295  & 102 & - & 4096  & 486\\
        24  & - & 5794  & 467  & 144 & - & 8192  & 768\\
        26  & - & 11586 & 740  & 202 & - & 16384 & 1221 \\
        28  & - & 23171 & 1174 & 285 & - & 32768 & 1938 \\
        30  & - & 46342 & 1862 & 403 & - & 65536 & 3072 \\
        32  & - & 92683 & 2955 & 569 & - & 131072  & 4878 \\
        34  & - & 185365  & 4690 & 803 & - & 262144  & 7743 \\
        36  & - & 370729  & 7444 & 1135  & - & 524288  & 12288\\
        38  & - & 741456  & 11816  & 1605  & - & 1048576 & 19506\\
        40  & - & - & 18756  & 2268  & - & - & 30966\\
        42  & - & - & 29773  & 3207  & - & - & 49152\\
        44  & - & - & 47261  & 4535  & - & - & 78024\\
        46  & - & - & 75021  & 6413  & - & - & 123858 \\
        48  & - & - & 119088 & 9068  & - & - & 196608 \\ \bottomrule
        \end{tabular}
    }
\end{table}

%% file: conclusion.tex
\section{Conclusion}
\label{conclusion}

In this work, we proposed equality operators for constant-weight codewords. We showed how these operators are up to 10 times faster than folklore equality operators. Furthermore, we proposed constant-weight PIR, a PIR protocol using equality operators which is an approach that was previously assumed to be impractical.
We showed how the communication and computation cost of constant-weight PIR grows at a slower rate compared to SealPIR and MulPIR, respectively. 
Furthermore, we showed how constant-weight PIR is extended to keyword PIR to be the first practical, single-round, single-server keyword PIR protocol. We provided a detailed analysis of effect of a large domain on the runtime of constant-weight keyword PIR and discussed how it can be used for applications such as private file retrieval.

%% file: appendix.tex
\section{Mappings to Constant-weight Codewords}
\label{sec:more-mappings}

In this section, we propose additional techniques to map elements to constant-weight codewords. As a reminder, the goal is for the mapping (and inverse mapping) procedure to be efficient and less expensive than storing an equivalence table.

\paragraph{Perfect Mapping.} The perfect mapping was described in \Cref{constructions}. Since the mapping is one-on-one, there also exists an inverse mapping which is described in \Cref{alg:inverse-mapping-perfect}. Similar to the mapping, the complexity of the inverse mapping procedure is $O(m+k)$.

\begin{algorithm}[!ht]
	 \caption[]{\textsc{Inverse Perfect Mapping}}
	 \label{alg:inverse-mapping-perfect}
	 \begin{flushleft}
		 \textbf{Input:} $y \in CW(m,k)$
	 \end{flushleft}
	 	\vspace{-3mm}
	 \begin{algorithmic}[1]
	 	\State $x = 0$
	 	\State $h = 1$
	 	\For {$m' \in [m]$}
		 	\If {$y[m']=1$} \\
				\hspace{10mm}$x = x + \binom{m'}{h}$ \\
				\hspace{10mm}$h = h+1$
			\EndIf
	 	\EndFor
	 	\vspace{-3mm}
	 \end{algorithmic}
	 \begin{flushleft}
    	 \textbf{Output:}  $x\in\NN_0$
	 \end{flushleft}
	 	\vspace{-3mm}
\end{algorithm}

The perfect mapping also preserves the order between the mapped elements. This is useful in applications where it is important to preserve the ordering of elements in the domain, e.g., comparison operators.

\paragraph{Lossy Mapping.}
In some cases, we may need to map elements of some large domain to constant-weight codewords but the size of the domain is too large to assign a distinct codeword to each element. Recall that if $S$ is the domain, the code length, $m$, needs to be chosen such that $\binom{m}{k} \geq |S|$ which results in a prohibitively large $m$.

To address this issue, we propose a lossy mapping inspired by Bloom filters. The procedure for the lossy mapping is given in \Cref{alg:lossy-mapping}.

\begin{algorithm}[!ht]
	 \caption[]{\textsc{Lossy Mapping}}
	 \label{alg:lossy-mapping}
	 \begin{flushleft}
		 \textbf{Parameters:} Series of uniformly random hash functions $(H_i:S\mapsto [m])_{i\in \NN}$ \\
		 \textbf{Input:} $x \in S, m, k\in \NN$
	 \end{flushleft}
	 	\vspace{-3mm}
	 \begin{algorithmic}[1]
		\State $cnt \leftarrow 0$
		\State $i \leftarrow 1$
	 	\State $y\leftarrow 0^m$
	 	\While {$cnt < k$}
	 		\State $m'=H_i(x)$
		 	\If {$y[m'] = 0$} \\
				\hspace{10mm}$y_[m']=1$ \\
				\hspace{10mm}$cnt = cnt + 1$
			\EndIf
			\State $i = i+1$
	 	\EndWhile
	 \end{algorithmic}
	 	\vspace{-3mm}
	 \begin{flushleft}
    	 \textbf{Output:}  $y\in CW(m,k)$
	 \end{flushleft}
	 	\vspace{-3mm}
\end{algorithm}

Based on the definition, a probability exists that unequal elements of the domain are mapped to the same codeword which is formalized in the following theorem.

\begin{theorem}\label{thm:lossy-coll-prob}
In \Cref{alg:lossy-mapping}, assume $(H_i:S\mapsto [m])_{i\in \NN}$ is a series of uniformly random hash functions and $M_{m,k}(x)$ is the output of the algorithm for input $x$, $m$, and $k$ with $(H_i)_{i\in \NN}$ as the parameters. For two randomly chosen elements $x,y \in S$ such that $x\neq y$,
\begin{align}
    \mathbb{P}\left[M_{m,k}(x) = M_{m,k}(y)\right] = \frac{1}{\binom{m}{k}}.
\end{align}
\end{theorem}

\begin{proof}
To prove this theorem, it suffices to prove that for any given codeword in the range of $M_{m,k}(x)$ such as $c$,
$$
    \mathbb{P}\left[M_{m,k}(x) = c\right] = \frac{1}{\binom{m}{k}}.
$$
We prove this by induction over $k$. For $k=1$, it is easy to see that 
$$
    \mathbb{P}\left[M_{m,1}(x) = c\right] = \frac{1}{m}
$$ for any $c\in Range(M_{m,1}(x))$.

Let $I(c)$ denote the positions in the codeword $c$ where the bit is set to one.
For $k>1$, the probability that $H_1(x)\in I(c)$ is equal to $\frac{k}{m}$. By induction, the probability that set of the next $k-1$ distinct outputs in the series $(H_i(x))_{i\geq 2}$ is equal to $I(c) - \{H_1(x)\}$ is equal to $\frac{1}{\binom{m-1}{k-1}}$. Hence
$$
    \mathbb{P}\left[M_{m,k}(x) = c\right] = \frac{k}{m}\frac{1}{\binom{m-1}{k-1}} = \frac{1}{\binom{m}{k}}.
$$
\end{proof}

Due to the lossy nature of the mapping, an inverse mapping is not available for the lossy mapping.

\section{Correctness of \Cref{alg:oblivious-expand}}
\label{appendix-expansion}

\begin{theorem}
The output of \Cref{alg:oblivious-expand} is identical to that of \Cref{alg:oblivious-expand-sealpir}.
\end{theorem}

\begin{proof}
    To prove the correctness of the oblivious expansion in \Cref{alg:oblivious-expand}, we prove it is equivalent to the oblivious expansion of SealPIR, shown in \Cref{alg:oblivious-expand-sealpir}.
    Also, let $\texttt{Sub}$ denote the substitution operation 
    For this, we prove that line 4--7 of \Cref{alg:oblivious-expand-sealpir} is equivalent to line 7--12 of \Cref{alg:oblivious-expand}.

    In \Cref{alg:oblivious-expand-sealpir}, denote $cts[b]$ on line 4 by $m(x)$ for simplicity. By executing lines 4 to 7, of the protocol, we can see that the new values for $cts[b]$ and $cts[b+2^a]$ are
    \begin{align*}
        cts[b]      &\leftarrow m(x) + \texttt{Sub}_{N/2^a+1}(m(x)) \\
        cts[b+2^a]  &\leftarrow x^{-2^a} \cdot m(x) + \texttt{Sub}_{N/2^a+1}(x^{-2^a}\cdot m(x))
    \end{align*}
    
    Similarly for \Cref{alg:oblivious-expand} and denoting $cts[b]$ on line 7 as $m(x)$, by executing lines 7 to 12, the new values for $cts[b]$ and $cts[b+2^a]$ are
    
    \begin{align*}
        cts[b]      &\leftarrow m(x) + \texttt{Sub}_{N/2^a+1}(m(x))\\
        cts[b+2^a]  &\leftarrow x^{-2^a} \cdot m(x) - x^{-2^a} \cdot \texttt{Sub}_{N/2^a+1}(m(x))
    \end{align*}
    
    So $cts[b]$ gets the same value after both protocols. To show that $cts[b+2^a]$ also gets the same value, it suffices to show that $\texttt{Sub}_{N/2^a+1}(x^{-2^a}\cdot m(x)) = - x^{-2^a} \cdot \texttt{Sub}_{N/2^a+1}(m(x))$ which can be proven as follows:
    \begin{align*}
        \texttt{Sub}_{N/2^a+1}(x^{-2^a}\cdot m(x))=&\ {(x^{N/2^a+1})}^{-2^a}\cdot m(x^{N/2^a+1})\\
        =&\ {x^{-N - 2^a}}\cdot m(x^{N/2^a+1})\\
        =&\ {-x^{- 2^a}}\cdot m(x^{N/2^a+1})\\
        =&\ - x^{-2^a} \cdot \texttt{Sub}_{N/2^a+1}(m(x))
    \end{align*}
\end{proof}

\section{Runtimes for Parallelized Operators}
\label{sec:parallel-plain}
Runtimes for parallelized plain operators are given in \Cref{tab:eq-parallel}. The runtimes in this table all have at most a 2 times speedup compared to the non-parallel version of the corresponding operator. The speedup for the folklore operator does not differ substantially from the speedup of the constant-weight operators.

Runtimes for parallel arithmetic operators are also given in \Cref{tab:eq-parallel}. Unlike the parallel operators, there is a substantial difference in the speedup that the folklore and constant-weight operators gain from parallelization. The folklore operator gains at most a 2 times speedup whereas the folklore operators gains up to a 10 fold speedup. 

\input{tables/eq-parallel}

\section{Detailed Runtimes of the Unary Approach, SealPIR and MulPIR}
\label{sec:other-approach}
The unary approach occurs when $k=1$ in constant-weight PIR, or when $d=1$ in SealPIR and MulPIR. In this approach, the selection vector in its entirety is communicated over the network.
In the unary approach, no expensive homomorphic operations such as homomorphic multiplications are performed. There is also no layered encryption as done in SealPIR. Hence, the server time is smaller than other protocols shown in this work. However, since the size of the selection vector is on the order of the number of rows in the database, the upload cost rises quickly as the number of rows grows. The upload cost becomes impractical very early, hence it is not a suitable solution for databases with a large number of rows.

We also provide numbers for SealPIR and MulPIR for payload of one plaintext in \Cref{tab:packed-pir-all}.

\input{tables/pir-packed-complete}

%% file: tables/eq-parallel.tex
\begin{table}[!ht]
    \centering
    \caption{Runtimes for plain and arithmetic equality operators in milliseconds when run in parallel. Dashes indicate cases where the ciphertext was undecryptable due to homomorphic noise. $k$ and $m$ denote the Hamming weight and constant-weight code length, respectively.}
    \label{tab:eq-parallel}
    \resizebox{\columnwidth}{!}{
        \begin{tabular}{c|cccccccc}\toprule
        \multicolumn{9}{c}{\textbf{Plain Operators}} \\
        \toprule
         & $n$ & $2^{8}$& $2^{16}$ & $2^{32}$ & $2^{64}$ & $2^{128}$ & $2^{256}$ & $2^{512}$ \\ \bottomrule \toprule
        \multirow{4}{*}{\begin{tabular}[c]{@{}c@{}}Plain\\Folklore\end{tabular}} & $\ell$ & 8& 16 & 32 & 64& 128 & 256 & 512 \\ 
         & Mult Depth & 3& 4& 5& 6 & 7 & 8 & 9 \\     
         & $N=8192$ & 0.20 & 0.25 & - & - & - & - & - \\
         & $N=16384$& 0.74 & 0.96 & 1.3 & 1.9 & 2.7 & 4.3 & 7.6 \\ \bottomrule \toprule
        \multirow{5}{*}{\begin{tabular}[c]{@{}c@{}}Plain\\Constant-\\weight\\$k=\log_2 n$\end{tabular}} & \begin{tabular}[c]{@{}c@{}}$k$\end{tabular} &  8& 16 & 32 & 64& 128 & 256 & 512\\
         & Mult Depth & 3& 4& 5& 6 & 7 & 8 & 9 \\
         & \begin{tabular}[c]{@{}c@{}}$m$\end{tabular}& 12 & 22 & 43 & 85& 168 & 334 & 665\\
         & $N=8192$ & 0.18 & 0.28 & - & - & - & - & - \\
         & $N=16384$& 0.58 & 1.0 & 1.2 & 1.9 & 2.6 & 4.1 & 6.9 \\ \bottomrule \toprule
        \multirow{5}{*}{\begin{tabular}[c]{@{}c@{}}Plain\\Constant-\\weight\\$k=\frac{1}{2}\log_2 n$\end{tabular}} & \begin{tabular}[c]{@{}c@{}}$k$\end{tabular} & 4& 8& 16 & 32& 64& 128 & 256\\
         & Mult Depth &  2& 3& 4& 5 & 6 & 7 & 8 \\
         & \begin{tabular}[c]{@{}c@{}}$m$\end{tabular}& 11 & 19 & 36 & 68& 132 & 261 & 517 \\
         & $N=8192$ & 0.15 & 0.18 & 0.24 & - & - & - & - \\
         & $N=16384$& 0.37 & 0.75 & 0.87 & 1.4 & 2.1 & 2.8 & 4.1\\ \bottomrule \toprule
        \multirow{6}{*}{\begin{tabular}[c]{@{}c@{}}Plain\\Constant-\\weight\\$k=\frac{1}{4}\log_2 n$\end{tabular}} & \begin{tabular}[c]{@{}c@{}}$k$\end{tabular} & 2& 4& 8& 16& 32& 64& 128 \\
         & Mult Depth & 1& 2& 3& 4 & 5 & 6 & 7 \\
         & \begin{tabular}[c]{@{}c@{}}$m$\end{tabular}& 24 & 37 & 64 & 117 & 221 & 427 & 838\\
         & $N=4096$ & 0.027 & - & - & - & - & - & -\\
         & $N=8192$ & 0.058 & 0.11 & 0.22 & 0.25 & - & - & -\\
         & $N=16384$& 0.18 & 0.5 & 0.76& 1.03 & 1.3 & 1.8 & 2.6 \\  \bottomrule \toprule
        \multirow{6}{*}{\begin{tabular}[c]{@{}c@{}}Plain\\Constant-\\weight\\$k=\frac{1}{8}\log_2 n$\end{tabular}} & \begin{tabular}[c]{@{}c@{}}$k$\end{tabular} & 1& 2& 4& 8 & 16& 32& 64\\
         & Mult Depth & 0& 1& 2& 3 & 4 & 5 & 6 \\
         & \begin{tabular}[c]{@{}c@{}}$m$\end{tabular}& 256& 363& 569& 968 & 1749& 3290& 6349\\
         & $N=4096$ & \textbf{0.0001} & \textbf{0.028} & - & - & - & - & -\\
         & $N=8192$ & 0.0005 & 0.067 & \textbf{0.14} & \textbf{0.22} & \textbf{0.27} & - & -\\
         & $N=16384$& 0.002 & 0.2 & 0.53 & 0.73 & 1.1 & \textbf{1.4} & \textbf{1.8}\\ \bottomrule
        \end{tabular}
    }

    \resizebox{\columnwidth}{!}{
        \begin{tabular}{c|cccccccc} \toprule
        \multicolumn{9}{c}{\textbf{Arithmetic Operators}} \\
        \toprule
        & $n$ & $2^{8}$& $2^{16}$ & $2^{32}$ & $2^{64}$ & $2^{128}$ & $2^{256}$ & $2^{512}$ \\ \bottomrule \toprule
        \multirow{4}{*}{\begin{tabular}[c]{@{}c@{}}Plain\\Folklore\end{tabular}} & $\ell$ & 8& 16 & 32 & 64& 128 & 256 & 512 \\ 
        & Mult Depth & 3& 4& 5& 6 & 7 & 8 & 9 \\     
        & $N=8192$ & 0.43 & - & - & - & - & - & -\\
        & $N=16384$ & 1.7 & 3.1 & 5.6 & 10 & 20 & 38 & 74 \\ \bottomrule \toprule
        \multirow{5}{*}{\begin{tabular}[c]{@{}c@{}}Arithmetic\\Constant-\\weight\\$k=\log_2 n$\end{tabular}} & $k$ & 8 & 16 & 32 & 64 & 128 & 256 & 512 \\
        & Mult Depth & 4 & 5 & 6 & 7 & 8 & 9 & 10 \\
        & $m$ & 12 & 22 & 43 & 85 & 168 & 334 & 665\\
        & $N=8192$  & 0.29 & - & - & - & - & - & -\\
        & $N=16384$  & 1.0 & 1.4 & 2.0 & 2.7 & 4.4 & 8.2 & 14 \\ \bottomrule \toprule
        \multirow{5}{*}{\begin{tabular}[c]{@{}c@{}}Arithmetic\\Constant-\\weight\\$k=\frac{1}{2}\log_2 n$\end{tabular}} & $k$ & 4 & 8 & 16 & 32 & 64 & 128 & 256\\
        & Mult Depth& 3 & 4 & 5 & 6 & 7 & 8 & 9 \\
        & $m$ & 11 & 19 & 36 & 68 & 132 & 261 & 517\\
        & $N=8192$ & \textbf{0.22} & \textbf{0.36} & - & - & - & - & -\\
        & $N=16384$ & 0.84 & 1.1 & 1.6 & \textbf{2.6} & \textbf{4} & \textbf{6.5} & \textbf{9.5}\\ \bottomrule \toprule
        \multirow{5}{*}{\begin{tabular}[c]{@{}c@{}}Arithmetic\\Constant-\\weight\\$k=\frac{1}{4}\log_2 n$\end{tabular}} & $k$ & 2 & 4 & 8 & 16 & 32 & 64 & 128\\
        & Mult Depth & 2 & 3 & 4 & 5 & 6 & 7 & 8\\
        & $m$ & 24 & 37 & 64 & 117 & 221 & 427 & 838 \\
        & $N=8192$ & 0.29 & 0.42 & - & - & - & - & -\\
        & $N=16384$ & 0.70 & 1.0 & \textbf{1.5} & 2.8 & 4.2 & 7.1 & 12 \\ \bottomrule \toprule
        \multirow{6}{*}{\begin{tabular}[c]{@{}c@{}}Arithmetic\\Constant-\\weight\\$k=\frac{1}{8}\log_2 n$\end{tabular}} & $k$ & 1 & 2 & 4 & 8 & 16 & 32 & 64 \\
        & Mult Depth & 1 & 2 & 3 & 4 & 5 & 6 & 7 \\
        & $m$ & 256 & 363 & 569 & 968 & 1749 & 3290 & 6349 \\
        & $N=4096$ & 0.44 & - & - & - & - & - & -\\
        & $N=8192$ & 0.81 & 1.1 & 1.6 & - & - & - & -\\
        & $N=16384$ & 3.0 & 4.5 & 7.0 & 12 & 20 & 37 & 73 \\ 
        \bottomrule \end{tabular}
    }

\end{table}

%% file: tables/pir-packed-complete.tex
\begin{table}[!ht]
    \centering
    \caption{Runtime of PIR protocols for a response size of one plaintext. Runtimes are in seconds and an average of 10 runs. *This parameter set did not produce a decryptable result}
    \label{tab:packed-pir-all}
    \resizebox{\columnwidth}{!}{%
    \begin{tabular}{@{\extracolsep{4pt}}ccccccc@{}}
    \toprule
    &&& \multicolumn{4}{c}{\textbf{Time (ms)}}\\ \cline{4-7}
    \begin{tabular}[c]{@{}c@{}}\textbf{\# of}\\\textbf{Rows}\end{tabular} &
    \begin{tabular}[c]{@{}c@{}}\textbf{DB Size}\\\textbf{(MB)}\end{tabular} &
    \begin{tabular}[c]{@{}c@{}}\textbf{Code}\\\textbf{Length}\end{tabular} & \begin{tabular}[c]{@{}c@{}}\textbf{Expansion}\end{tabular} & \begin{tabular}[c]{@{}c@{}}\textbf{Sel. Vec.}\\\textbf{Calculation}\end{tabular} & \begin{tabular}[c]{@{}c@{}}\textbf{Inner}\\\textbf{Product}\end{tabular} & \begin{tabular}[c]{@{}c@{}}\textbf{Total Server}\end{tabular}\\
    \bottomrule
    \toprule

    \multicolumn{7}{c}{Folklore, $N=8192$ (Query = 216 KB, Response = 106 KB)} \\
    \midrule
    256 &	8 &	5 &	0.06 &	58 &	0.9 & 60 \\
    512$^{*}$ & 9 &	10 & 0.1 & 130 & 1.7 & 130 \\
    \bottomrule
    \toprule
    \multicolumn{7}{c}{Folklore, $N=16384$ (Query = 913 KB, Response = 224 KB)} \\
    \midrule
    512 &	21 &	9 &	0.8 &	650 &	7.4 &	660 \\
    1024 &	42 &	10 &	0.8 &	1500 &	14 &	1500 \\
    2048 &	84 &	11 &	0.8 &	3300 &	29 &	3300 \\
    4096 &	170 &	12 &	0.8 &	7200 &	56 &	7200 \\
    8192 &	340 &	13 &	0.8 &	16000 &	120 &	16000 \\
    16384 &	670 &	14 &	0.8 &	35000 &	250 &	35000 \\
    \bottomrule
    \toprule
    \multicolumn{7}{c}{Unary, $N=4096$ (Response = 46 KB)} \\
    \midrule
    256 & 2.6 &	256 &	0.5 &	0.009 &	0.2 &	0.8 \\
    512 & 5.2 &	512 &	1 &	0.02 &	0.4 &	1.4 \\
    1024 & 10 &	1024 &	1.9 &	0.05 &	0.8 &	2.8 \\
    2048 & 21 &	2048 &	3.8 &	0.2 &	1.7 &	5.7 \\
    4096 & 42 &	4096 &	7.7 &	0.5 &	3.3 &	11 \\
    8192 & 84 &	8192 &	15 &	1.9 &	6.4 &	24 \\
    16384 & 170 &	16384 &	30 &	6.7 &	13 &	49 \\
    32768 & 340 &	32768 &	59 &	23 &	25 &	110 \\
    65536 & 670 &	65536 &	120 &	87 &	52 &	260 \\
    131072 & 1300 &	131072 &	240 &	340 &	110 &	680 \\
    \bottomrule
    \toprule
    \multicolumn{7}{c}{Constant-weight $k=2, N=8192$, (Query = 216 KB, Response = 106 KB)} \\
    \midrule
    256 & 5.2 &	24 &	0.3 &	8.3 &	0.9 &	9.7 \\
    512 & 10 &	33 &	0.5 &	17 &	1.7 &	19 \\
    1024 & 21 &	46 &	0.5 &	33 &	3.5 &	38 \\
    2048 & 42 &	65 &	1 &	67 &	6.9 &	75 \\
    4096 & 84 &	92 &	1 &	130 &	13 &	150 \\
    8192 & 170 &	129 &	2 &	270 &	27 &	300 \\
    16384 & 340 &	182 &	2 &	540 &	55 &	600 \\
    32768 & 670 &	257 &	5 &	1100 &	110 &	1200 \\
    65536 & 1300 &	363 &	5 &	2300 &	230 &	2500 \\
    
    \bottomrule
    \toprule
    \multicolumn{7}{c}{SEALPIR $d=2, N=4096$ (Query = 61.4 KB, Response = 307 KB)} \\
    \midrule
    512 & 4.98 & 46 & - & - & - & 0.34 \\
    1024 & 9.96 & 64 & - & - & - & 0.46 \\
    2048 & 19.9 & 92 & - & - & - & 0.80 \\
    4096 & 39.8 & 128 & - & - & - & 1.2 \\
    8192 & 79.6 & 182 & - & - & - & 2.2 \\
    16384 & 159 & 256 & - & - & - & 3.7 \\
    32768 & 318 & 364 & - & - & - & 7.0 \\
    65536 & 637 & 512 & - & - & - & 12 \\
    131072 & 1275 & 726 & - & - & - & 24 \\
    262144 & 2550 & 1024 & - & - & - & 50 \\
    524288 & 5100 & 1450 & - & - & - & 100 \\
    1048576 & 10200 & 2048 & - & - & - & 200 \\
    2097152 & 20401 & 2898 & - & - & - & 430 \\
    \bottomrule
    \toprule
    \multicolumn{7}{c}{MulPIR $d=2, N=8192$ (Query = 122 KB, Response = 119 KB)} \\
    \midrule
    256 & 4.98 & 32 & - & - & - & 2.3 \\
    512 & 9.96 & 46 & - & - & - & 4.1 \\
    1024 & 19.9 & 64 & - & - & - & 6.8 \\
    2048 & 39.8 & 92 & - & - & - & 12 \\
    4096 & 79.6 & 128 & - & - & - & 22 \\
    8192 & 159 & 182 & - & - & - & 44 \\
    16384 & 318 & 256 & - & - & - & 83 \\
    32768 & 637 & 364 & - & - & - & 160 \\
    65536 & 1275 & 512 & - & - & - & 320 \\
    131072 & 2550 & 726 & - & - & - & 630 \\
    262144 & 5100 & 1024 & - & - & - & 1200 \\
    524288 & 10200 & 1450 & - & - & - & 2500 \\
    \bottomrule
    \end{tabular}%
    }
    \end{table}